\documentclass[11pt, a4paper]{amsart}
\usepackage[lmargin=34mm,rmargin=34mm,tmargin=25mm,bmargin=25mm]{geometry}
\renewcommand{\baselinestretch}{1.2}
\usepackage{latexsym}
\usepackage{amsfonts}
\usepackage{amsmath}
\usepackage{amsthm}
\usepackage{amssymb}
\usepackage{pst-node}
\usepackage[utf8]{inputenc}
\usepackage{fontenc}
\usepackage{graphicx}
\usepackage{latexsym}
\usepackage{color}

\usepackage{epsfig}
\usepackage{algorithm}
\usepackage{algorithmic}
\usepackage{enumerate}

\newcommand{\mc}{\mathcal}
\newcommand{\DEF}{\sl}

\newcommand{\tw}{\mathrm{tw}} % tree-width
\newcommand{\wt}{\widetilde}

\newcommand{\trc}{t} % top residual cost
\newcommand{\brc}{b} % bottom residual cost

\newcommand{\mxc}{q} % the old t

\newtheorem{theorem}{Theorem}[section]
\newtheorem{lemma}[theorem]{Lemma}
\newtheorem{claim}[theorem]{Claim}
\newtheorem{proposition}[theorem]{Proposition}

\title{Hitting Diamonds and Growing Cacti}
\thanks{
A preliminary version of the paper will appear in the proceedings of the 
14th Conference on Integer Programming and Combinatorial Optimization (IPCO 2010). 
This work was supported by the ``Actions de Recherche Concert\'ees'' (ARC) fund of the ``Communaut\'e fran\c{c}aise de Belgique''. G.J. is a Postdoctoral Researcher of the ``Fonds National de la Recherche Scientifique'' (F.R.S.--FNRS). 
This work was done while 
U.P.~was at  D\'epartement de Math\'ematique - Universit\'e Libre de Bruxelles as a Postdoctoral Researcher of the F.R.S.--FNRS}

\author{Samuel Fiorini}
\address{\newline D\'epartement de Math\'ematique
\newline Universit\'e Libre de Bruxelles
\newline Brussels, Belgium}
\email{sfiorini@ulb.ac.be}

\author{Gwena\"el Joret}
\address{\newline D\'epartement d'Informatique
\newline Universit\'e Libre de Bruxelles
\newline Brussels, Belgium}
\email{gjoret@ulb.ac.be}

\author{Ugo Pietropaoli}
\address{\newline Dipartimento di Ingegneria dell'Impresa 
\newline Universit\`a di Roma ``Tor Vergata'' 
\newline Rome, Italy}
\email{pietropaoli@disp.uniroma2.it}  

\date{\today}

\begin{document}

\maketitle

%%%%%%%%%%%%%%%%%%%%%%%%%%%%%%%%%%%%%%%%%%%%%%%%%%%%%%%%%%%
\begin{abstract} 
We consider the following NP-hard problem: in a weighted graph, find a minimum cost set of vertices whose removal leaves a graph in which no two cycles share an edge. 
We obtain a constant-factor approximation algorithm, based on the primal-dual method. Moreover, we show that the integrality gap of the natural LP relaxation of the problem is $\Theta(\log n)$, where $n$ denotes the number of vertices in the graph.
\end{abstract} 

{\small {\sc Keywords}: Approximation algorithm, primal-dual method, covering problem.} 

%%%%%%%%%%%%%%%%%%%%%%%%%%%%%%%%%%%%%%%%%%%%%%%%%%%%%%%%%%%
\section{Introduction}
\label{intro}

Graphs in this paper are finite, undirected, and may contain parallel edges but no loops.
We study the following combinatorial optimization problem: given a vertex-weighted graph, remove 
a minimum cost subset of vertices so that all the cycles in the resulting graph 
are edge-disjoint. 
We call this problem  the {\DEF diamond hitting set problem}, because it
is equivalent to covering all subgraphs which are diamonds with a minimum
cost subset of vertices, where a {\DEF diamond} is
any subdivision of the graph consisting of three parallel edges.

The diamond hitting set problem can be thought of as a 
generalization of the vertex cover and feedback vertex set problems: 
Suppose you wish to remove a minimum cost subset of vertices so 
that the resulting graph has no pair of vertices linked by $k$ internally disjoint paths. 
Then, for $k=1$ and $k=2$, this is respectively the vertex cover problem and feedback vertex set problem, 
while for $k=3$ this corresponds to the diamond hitting set problem.

It is well-known that both the vertex cover and feedback vertex set problems
admit constant-factor approximation algorithms\footnote{A $\rho$-approximation algorithm for a minimization problem is an algorithm that runs in polynomial time and outputs a feasible solution whose cost is no more than $\rho$ times the cost of the optimal solution. The number $\rho$ is called the approximation factor.}. 
Hence, it is natural to ask whether the same is true for the 
diamond hitting set problem. The main contribution of this paper is a positive answer to this question.

\subsection{Background and Related Work}

Although there exists a simple $2$-approximation algorithm for the vertex cover problem, there is strong evidence that approximating the problem with a factor of $2 - \varepsilon$ might be hard, for every 
$\varepsilon > 0$ \cite{KR08}. 
It should be noted that the feedback vertex set and diamond hitting set problems are at least
as hard to approximate as the vertex cover problem, in the sense that the existence of a 
$\rho$-approximation algorithm for one of these two problems implies the existence of a $\rho$-approximation algorithm for the vertex cover problem, where $\rho$ is a constant.

Concerning the feedback vertex set problem, the first approximation algorithm is due to Bar-Yehuda, Geiger, Naor, and Roth \cite{bgnr98} and its approximation factor is $O(\log n)$. Later, $2$-approximation algorithms have been proposed by Bafna, Berman, and Fujito \cite{bbf99}, and Becker and Geiger \cite{bg96}. Chudak, Goemans, Hochbaum and Williamson~\cite{cghw98} showed that these algorithms can be seen as deriving from the primal-dual method (see for instance~\cite{ps82, gw97}). Starting with an integer programming formulation of the problem, these algorithms simultaneously construct a feasible integral solution and a feasible dual solution of the linear programming relaxation, such that the values of these two solutions are within a constant factor of each other.

These algorithms also lead to a characterization of the integrality gap\footnote{The integrality gap of an integer programming formulation is the worst-case ratio between the optimum value of the integer program and the optimum value of its linear relaxation.} of two different integer programming formulations of the problem, as we now explain. Let $\mathcal{C}(G)$ denote the collection of all the cycles $C$ of $G$. 
A natural integer programming formulation for the feedback vertex set problem is as follows:
\begin{alignat}{3}
\nonumber \mbox{Min } & \sum_{v \in V(G)} c_v\,x_v & \\
\label{fvscover} \mbox{s.t. } & \sum_{v \in V(C)} x_v \geqslant 1  & \qquad \forall C \in \mathcal{C}(G)\\
\nonumber & x_v \in \{0,1\} & \qquad \forall v \in V(G).
\end{alignat}
(Throughout, $c_{v}$ denotes the (non-negative) cost of vertex $v$.)
The algorithm of Bar-Yehuda et al.~\cite{bgnr98} implies that the integrality gap of this integer program is $O(\log n)$. Later, Even, Naor, Schieber, and Zosin \cite{ensz00} proved that its integrality gap is also $\Omega (\log n)$.

A better formulation has been introduced by Chudak et al.~\cite{cghw98}.
For $S \subseteq V(G)$, denote by $E(S)$ the set of the edges of $G$ having both ends in $S$, by $G[S]$ the subgraph of $G$ induced by $S$, and by $d_S(v)$ the degree of $v$ in $G[S]$. Then, the following is a formulation for the feedback vertex set problem:
\begin{alignat}{3}
\nonumber \mbox{Min } & \sum_{v \in V(G)} c_v\,x_v & \\
\label{fvsspar} \mbox{s.t. } & \sum_{v \in S} (d_S(v)-1) x_v \geqslant |E(S)| - |S| +1 & \qquad \forall S \subseteq V(G) : E(S) \neq \varnothing\\
\nonumber & x_v \in \{0,1\} & \qquad \forall v \in V(G).
\end{alignat}
Chudak et al.~\cite{cghw98} showed that the integrality gap of this integer program asymptotically equals $2$. Constraints (\ref{fvsspar}) derive from the simple observation that the removal of a feedback vertex set $X$ from $G$ generates a forest having at most $|G| - |X| - 1$ edges. Notice that the covering inequalities (\ref{fvscover}) are implied by (\ref{fvsspar}). 

\subsection{Contribution and Key Ideas}

First, we obtain a $O(\log n)$-approximation algorithm for the diamond hitting set problem, 
leading to a proof that the integrality gap of the natural LP formulation is $\Theta(\log n)$. Then, we develop a 9-approximation algorithm. Both the $O(\log n)$- and 9-approximation algorithm are based on the primal-dual method.   

Our first key idea is contained in the following observation: every simple graph of order $n$ and minimum degree at least $3$ contains a $O(\log n)$-size diamond. This directly yields a $O(\log n)$-approximation algorithm for the diamond hitting set problem, in the unweighted case. However, the weighted case requires more work. 

Our second key idea is to generalize constraints (\ref{fvsspar}) by introducing `sparsity inequalities', that enable us to derive a constant-factor approximation algorithm for the diamond hitting set problem: 
First, by using reduction operations, we ensure that every vertex of $G$ has at least three neighbors.
Then, if $G$ contains a diamond with at most $9$ edges, we raise the dual variable of the corresponding covering constraint. Otherwise, no such small diamond exists in $G$, and we can use this information to select the right sparsity inequality, and raise its dual variable. This inequality would not be valid in case $G$ contained a small diamond.

The way we use the non-existence of small diamonds is perhaps best explained
via an analogy with planar graphs: An $n$-vertex planar simple graph $G$
has at most $3n-6$ edges. However, if we know that $G$ has no small cycle,
then this upper bound can be much strengthened. (For instance, if $G$ is triangle-free then $G$ has at most $2n-4$ edges.) 

We remark that this kind of local/global trade-off did not appear in the work of Chudak et al.~\cite{cghw98} on the feedback vertex set problem, because 
the cycle covering inequalities are implied by their more general inequalities. 
In our case, the covering inequalities and the sparsity inequalities form two incomparable classes of inequalities, and examples show that the sparsity inequalities alone are not enough to derive a constant-factor approximation algorithm. 

The paper is organized as follows. Preliminaries are given in Section \ref{prelim}. Then, in Section~\ref{reductions}, we define some reduction operations that allow us to obtain graphs having some desirable properties. Next, in Section \ref{logn_unweighted}, we deal with the unweighted version of the diamond hitting set problem and provide a simple $O(\log n)$-approximation algorithm. In Section \ref{logn_weighted}, we turn to the weighted version of the problem. We present a $O(\log n)$-approximation algorithm and we prove that the integrality gap of the natural formulation of the problem is $\Theta(\log n)$. 
Finally, in Section \ref{const_weighted}, we introduce the sparsity inequalities, analyze their strength and obtain a $9$-approximation algorithm. 

\section{Preliminaries} 
\label{prelim}

We refer the reader to Diestel~\cite{D05} for undefined terms and notations concerning graphs.
A {\DEF cactus} is a connected graph where each edge belongs to at most one cycle. 
Equivalently, a connected graph 
is a cactus if and only if each of its blocks is isomorphic to either $K_{1}$, $K_{2}$, or a cycle.
Thus, a connected graph is a cactus if and only if it does not contain a diamond as a subgraph.
A graph without diamonds is called a {\DEF forest of cacti} 
(see Figure~\ref{fig_forest_cacti} for an illustration).

\begin{figure}[htb]
\centering 
\includegraphics[scale=.9]{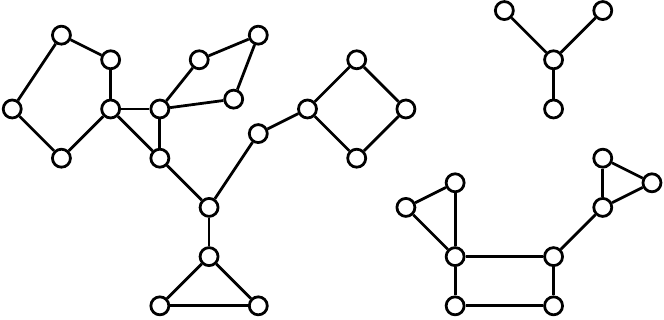}
\caption{A forest of cacti.}
\label{fig_forest_cacti}
\end{figure}

A {\DEF diamond hitting set}, or simply {\DEF hitting set}, of a graph is a subset of vertices that hits every diamond of the graph.  A {\DEF minimum hitting set} of a weighted graph is a hitting set of minimum total cost, and its cost is denoted by $OPT$. 

Let $\mathcal{D}(G)$ denote the collection of all diamonds contained in $G$. From the standard IP formulation for a covering problem, we obtain the following LP relaxation for the diamond hitting set problem:
\begin{alignat}{3}
\nonumber \mbox{Min } && \sum_{v \in V(G)} c_v\,x_v & \\
\label{ds_ineq} \mbox{s.t. } && \sum_{v \in V(D)} x_v &\geqslant 1  & \qquad \forall D \in \mathcal{D}(G)\\
\nonumber && x_v &\geqslant 0 & \qquad \forall v \in V(G).
\end{alignat}
We call inequalities (\ref{ds_ineq}) {\DEF diamond inequalities}. 

% ======================================================================
\section{Reductions}
\label{reductions}

In this section, we define two reduction operations on graphs: 
First, we define the `shaving' of an arbitrary graph, and then introduce 
a `bond reduction' operation for shaved graphs.

The aim of these two operations is to modify a given graph so that the following useful property holds: 
each vertex either has at least three distinct neighbors, or is incident to at least three parallel edges.

\subsection{Shaving a graph}

Let $G$ be a graph. Every block of $G$ is either isomorphic to $K_{1}$, $K_{2}$, a cycle, or
contains a diamond. Mark every vertex of $G$ that is included in a block containing a diamond.
The {\DEF shaving} of $G$ is the graph obtained by removing
every unmarked vertex from $G$. A graph is {\DEF shaved} if all its vertices belong to a block containing a diamond. Observe that, in particular, every endblock\footnote{
We recall that the block-graph of $G$ has the blocks of $G$ and the cutvertices of $G$ as vertices, a block and a cutvertex are adjacent if the former contains the latter. This graph is always acyclic. An endblock of $G$ is a vertex of the block-graph with degree at most one.} 
of a shaved graph contains a diamond.
See Figure~\ref{fig:shaving} for an illustration.

\begin{figure}[htb]
\centering    
\includegraphics[scale=.9]{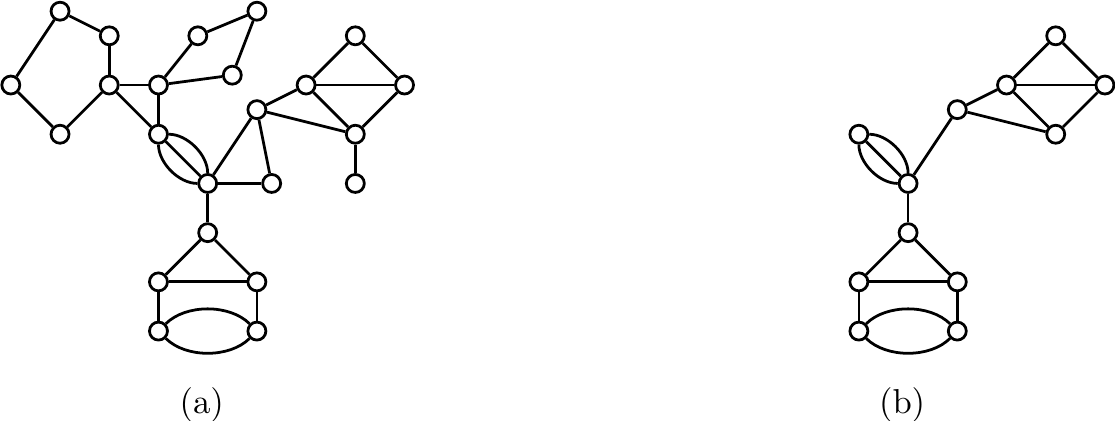}
\caption{\label{fig:shaving} (a) A graph $G$. (b) The graph obtained by shaving $G$.}   
\end{figure}

\subsection{Reducing a bond}

A {\DEF bond} of a graph $G$ is a connected subgraph $Q \subseteq G$ 
equipped with two distinguished vertices $v,w$ (called {\DEF ends})
satisfying the following requirements:
\begin{itemize}
\item $Q$ is a cactus with at least two blocks;
\item the block-graph of $Q$ is a path;
\item $v$ and $w$ belong to distinct endblocks of $Q$;
\item $v$ and $w$ are not adjacent in $Q$;
\item $Q - \{v,w\}$ is a non-empty component of $G - \{v, w\}$, and
\item $Q$ contains all the edges in $G$ between $\{v,w\}$ and $V(Q) - \{v, w\}$.
\end{itemize}
Observe that $Q$ is ``almost'' an induced subgraph of $G$, since $Q$ includes every edge of $G$ between vertices of $Q$, except those between $v$ and $w$ (if any).
The vertices in $V(Q) - \{v,w\}$ are said to be the {\DEF internal vertices} of $Q$.
The bond $Q$ is {\DEF simple} if $Q$ is isomorphic to a path, {\DEF double} otherwise.

Let $G$ be a shaved graph. A vertex $u$ of $G$ is {\DEF reducible} if $u$ has exactly two neighbors in $G$, and there are at most two parallel edges connecting $u$ to each of its neighbors. 
The {\DEF bond reduction} operation is defined as follows. Let $u$ be a reducible vertex and let
$Q_u$ be an inclusion-wise maximal bond of $G$ containing $u$, with ends $v$ and $w$. (Observe that such a bond exists by our hypothesis on $u$; moreover, it might not be unique.) Then, 
remove from $G$ every  internal vertex of $Q_{u}$, and add one or two edges
between $v$ and $w$, depending on whether $Q_u$ is simple or double. 
In the latter case, the two new parallel edges are said to be {\DEF twins}. 
See Figure~\ref{bondreduction} for an illustration of the operation. 
Observe that the resulting graph is also a shaved graph.

\begin{figure}[htb]
\centering    
\includegraphics{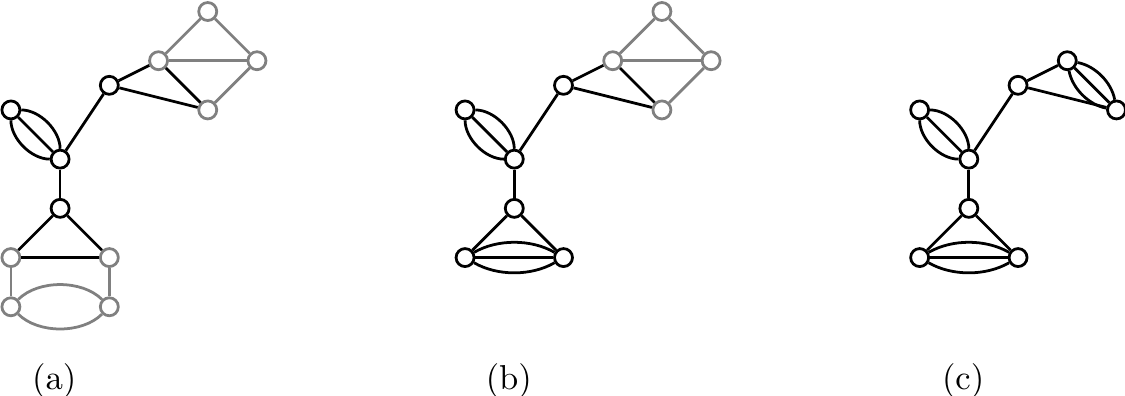}
\caption{\label{bondreduction} 
(a) A shaved graph $G$ with two maximal bonds (in grey).
(b) Reduction of the first bond.
(c) Reduction of the second bond. The graph is now reduced.}   
\end{figure}

A crucial property of the bond reduction operation is that, when applying it iteratively, we never include in the bond to be reduced any edge coming from previous bond reductions. 
This is proved in the following lemma.

\begin{lemma}
\label{lem:disjoint}
Let $G_i$ be a graph obtained from $G$ after 
applying $i$ bond reductions. Let $\bar E_i:= E(G_i) - E(G)$.
Let $v$ be a reducible vertex of $G_{i}$. Let $Q_v$ be a maximal bond of $G_i$ containing $v$.
Then, $E(Q_v) \cap \bar E_i = \varnothing$.
\end{lemma}
\begin{proof}
Arguing by contradiction, assume $E(Q_v) \cap \bar E_i \neq \varnothing$. 
Let $j$ be the maximum index such that $j<i$ and $Q_{v}$ contains an edge $e$ produced 
during the $j$th bond reduction.
Let $w$ be the vertex that has been reduced at iteration $j$, and denote by $Q_{w}$ the
corresponding bond.
Note that, if $Q_w$ is double, then $e$ and its twin edge $e'$ are both included in $Q_v$ (by construction). 
Now, replacing $e$ (and $e'$ if it exists) by $Q_{w}$ in $Q_{v}$ produces a bond of $G_{j}$
that includes $w$ and is strictly larger than $Q_{w}$. This contradicts the maximality of $Q_w$.
\end{proof}

A {\DEF reduced graph} $\wt G$ of $G$ is any graph obtained from $G$ by iteratively applying a bond reduction, as long as there is a reducible vertex (see Figure~\ref{bondreduction}). 
We remark that there is not necessarily a unique reduced graph of $G$ 
(consider for instance $K_{3}$ where two edges are doubled). 

As mentioned earlier, every reduced graph has the following desirable property: every vertex either has at least three distinct neighbors or is incident to at least three parallel edges.

% ======================================================================
\section{A $O(\log n)$-approximation algorithm in the unweighted case}
\label{logn_unweighted}

In this section, we deal with the unweighted version of the diamond hitting set problem. We first show that every reduced graph 
%%% OLD
%`reduced' graph (to be defined later) 
with $n$ vertices contains a diamond of size $O(\log n)$, and then provide a $O(\log n)$-approximation algorithm.

\subsection{Small diamonds in reduced graphs}

As a first step, we show that every simple graph with minimum degree at least $3$ contains a 
diamond of size $O(\log n)$.

\begin{lemma}
\label{lem:supercubic_implies_log_diamond}
Every simple $n$-vertex graph with minimum degree at least $3$ contains a diamond
of size at most $6\log_{3/2} n + 8$. Moreover, such a diamond can be found in polynomial time.
\end{lemma}
\begin{proof}
Let $G$ be a simple graph with $n$ vertices and minimum degree at least $3$, and let $u$ be an arbitrary vertex of $G$.
For $d\geqslant 0$, let $N_{d}$ be the set of vertices at distance $d$ from $u$; 
also, let $T_{d}$ be a spanning tree of the graph induced by $N_{0} \cup \cdots \cup N_{d}$
such that, for every vertex $v$ of $T_{d}$, the distance from $u$ to $v$ in $T_{d}$
is the same as in $G$. (Such a tree can be obtained by {\em breadth-first search}, for instance.)
Notice that the set of leaves of $T_{d}$ is exactly $N_{d}$.

For $d\geqslant 1$, color in white the vertices in $N_{d}$ that have exactly one neighbor
in $N_{d-1} \cup N_{d}$; in black those that have exactly two, 
and in red those that have at least three of them.

\begin{claim}
\label{claim:1}
Let $d \geqslant 1$. If $N_{d}$ contains a red vertex, then $G$ has a diamond of size at most $3d+2$.
The same is true if $N_{d}$ contains a black vertex in $N_{d}$ with a black neighbor in 
$N_{d-1}$, or two black vertices in $N_{d}$ that are not adjacent and have
a common neighbor in $N_{d-1}$.
\end{claim}
\begin{proof}
We only consider the case of a red vertex, the proof of the two other cases is similar and
left to the reader.
Let $v \in N_{d}$ be a red vertex.
Let $v_{1}$ be the unique neighbor of $v$ in the tree $T_{d}$. (Thus, $v_{1} \in N_{d-1}$.)
Let $v_{2}, v_{3}$ be two neighbors of $v$ in $N_{d-1} \cup N_{d}$ that are distinct from $v_{1}$.
Let $C_{2}$ and $C_{3}$ be the unique cycles in $T_{d} + vv_{2}$ and $T_{d} + vv_{3}$, respectively.
These two cycles are distinct and share the edge $vv_{1}$; hence, 
$C_{2} \cup C_{3}$ contains a diamond $D$.
Let $P_{i}$ ($i=1,2,3$) be the $u$-$v_{i}$ path in $T_{d}$, augmented with the edge $vv_{i}$.
Observe that $C_{2} \cup C_{3} \subseteq P_{1} \cup P_{2} \cup P_{3}$. 
Therefore, the size of $D$ is 
$$
||D|| \leqslant ||C_{2} \cup C_{3} || \leqslant ||P_{1}|| + ||P_{2}|| + ||P_{3}|| \leqslant d + (d+1) + (d+1) = 3d+2.
$$

\end{proof}

\begin{claim}
\label{claim:2}
Let $d \geqslant 2$. If $G$ has no diamond of size at most $3d+2$, then
$|N_{d}| \geqslant \frac32 |N_{d-2}|$.
\end{claim}
\begin{proof}
Assume that $G$ has no diamond of size at most $3d+2$. 
The claim is easily seen to hold if $d=2$, so we assume $d \geqslant 3$.

For $i \in \{1, \dots, d\}$,
let $w_{i}$ and $b_{i}$ be the number of white and black vertices in $N_{i}$, respectively.
By Claim~\ref{claim:1}, there is no red vertex in $N_{i}$; thus, $|N_{i}| = w_{i} + b_{i}$.
Let $p_{i}$ and $q_{i}$ be the number of black vertices in $N_{i}$ having
one and two neighbors in $N_{i-1}$, respectively; so, $b_{i} = p_{i} + q_{i}$.

Now, fix $i \in \{2, \dots, d\}$. (Thus, vertices in both $N_{i}$ and $N_{i-1}$ 
are colored in black and white.) 
Since $G$ has minimum degree at least $3$, every vertex in $N_{i-1}$ has at least one neighbor in $N_{i}$.
By Claim~\ref{claim:1}, every black vertex in $N_{i-1}$ has at least one white neighbor in $N_{i}$, showing
\begin{equation}
\label{eq:eq1}
w_{i} \geqslant b_{i-1}.
\end{equation}
The number of edges between $N_{i-1}$ and $N_{i}$ in $G$ 
is exactly $w_{i} + p_{i} + 2q_{i} =: m_{i}$.
On the other hand, we also have $m_{i} \geqslant 2w_{i-1} + b_{i-1}$, since
every white vertex in $N_{i-1}$ has at least two neighbors in $N_{i}$
and every black vertex in $N_{i-1}$ has at least one neighbor in $N_{i}$. 
Hence,
\begin{equation}
\label{eq:eq2}
w_{i} + p_{i} + 2q_{i} \geqslant 2w_{i-1} + b_{i-1}.
\end{equation}
By Claim~\ref{claim:1}, if a white vertex in $N_{i-1}$
is adjacent to two black vertices in $N_{i}$, then the latter vertices are adjacent.
Also, black vertices in $N_{i}$ can only be adjacent to white vertices in $N_{i-1}$, 
again by Claim~\ref{claim:1}. These two observations imply
\begin{equation}
\label{eq:eq3}
\frac12 w_{i-1} \geqslant q_{i}.
\end{equation}

Using Eq.~\eqref{eq:eq2} and~\eqref{eq:eq3}, we obtain
$$
|N_{i}| = w_{i} + p_{i} + q_{i}
\geqslant 2w_{i-1} + b_{i-1} - q_{i}
\geqslant \frac32 w_{i-1} + b_{i-1}.
$$
It follows
\begin{align*}
|N_{d}| &\geqslant \frac32 w_{d-1} + b_{d-1} \\
&= |N_{d-1}| + \frac12 w_{d-1} \\
&\geqslant \frac32 w_{d-2} + b_{d-2} + \frac12 w_{d-1} \\
&= |N_{d-2}| + \frac12 w_{d-2} + \frac12 w_{d-1} \\
&\geqslant |N_{d-2}| + \frac12 w_{d-2} + \frac12 b_{d-2} \\
&= \frac32 |N_{d-2}|,
\end{align*}
as claimed. (The last inequality follows from Eq.~\eqref{eq:eq1}.)
\end{proof}

Now, we may prove Lemma~\ref{lem:supercubic_implies_log_diamond}.
Let $d$ be the largest even integer such that $G$ has no diamond of size at most $3d + 2$.
If $d \leqslant 2$, then $G$ has a diamond with size at most 
$3(d+2) + 2 \leqslant 14 \leqslant 6\log_{3/2} n + 8$ (since $n \geqslant 4$).
Thus, assume $d\geqslant 4$. By Claim~\ref{claim:2},
\begin{align*}
n &\geqslant \sum_{i=0}^{d} |N_{i}| \\
&\geqslant |N_{0}| + |N_{2}| + \cdots + |N_{d-2}| + |N_{d}| \\
&\geqslant 1 + \frac32 + \cdots + \left(\frac32\right)^{d/2 -1} + \left(\frac32\right)^{d/2} \\
&\geqslant \left(\frac32\right)^{d/2},
\end{align*}
that is,
$$
d \leqslant 2 \log_{3/2} n.
$$
Therefore, $G$ has a diamond of size at most $3(d+2) + 2 \leqslant 6\log_{3/2} n + 8$.

To conclude, we note that the above proof is easily turned into a polynomial-time algorithm finding
the desired diamond.
\end{proof}

The same result holds for reduced graphs:

\begin{lemma}
\label{lem:reduced_implies_log_diamond}
Every reduced graph $\wt{G}$ contains a diamond of size at most $6\log_{3/2} |\wt{G}| + 8$.
\end{lemma}

\begin{proof}
We can assume that any two adjacent vertices of $\wt{G}$ are adjacent through at most $2$ parallel edges, since otherwise there is a diamond with three edges and the statement trivially holds. Thus, each vertex of $\wt{G}$ has at least three distinct neighbors. Let $\wt{G}'$ be the subgraph of $\wt{G}$ obtained by replacing every double edge by a simple edge. Clearly, $\wt{G}'$ is simple and has minimum degree at least $3$. Therefore, by Lemma \ref{lem:supercubic_implies_log_diamond}, $\wt{G}'$ has a diamond (which can be found in polynomial time) of size at most  $6\log_{3/2} |\wt{G}'| + 8$, thus at most $6\log_{3/2} |\wt{G}| + 8$. The result follows. 
\end{proof}

\subsection{The algorithm}

Our algorithm for the diamond hitting set problem on unweighted graphs is described
in Algorithm~\ref{algo:unweighted}.

\begin{algorithm}
\caption{\label{algo:unweighted}A $O(\log n)$-approximation algorithm for unweighted graphs.}
\begin{itemize}
\item $X \leftarrow \varnothing$
\item While $X$ is not a hitting set of $G$, repeat the following steps:
\begin{itemize}
\item Compute a reduced graph $\wt{G}$ of $G-X$
\item Find a diamond $\wt{D}$ in $\wt{G}$ of size at most $6\log_{3/2}|\wt{G}| + 8$ \hfill
(using Lemma \ref{lem:reduced_implies_log_diamond})
\item Include in $X$ all vertices of $\wt{D}$
\end{itemize}
\end{itemize}
\end{algorithm}

The algorithm relies on the simple fact that every hitting set of a reduced graph $\wt G$ of 
a graph $G$ is also a hitting set of $G$ itself.
The set of diamonds computed by the algorithm yields a collection $\mathcal{D}$ of pairwise vertex-disjoint diamonds in $G$. In particular, the size of a minimum hitting set is at least $|\mathcal{D}|$. For each diamond in $\mathcal{D}$, at most $6\log_{3/2} n + 8$ vertices were added to
the hitting set $X$. Hence, the approximation factor of the algorithm is $O(\log n)$.

% ======================================================================

\section{A $O(\log n)$-approximation algorithm}
\label{logn_weighted}

The present section is devoted to a $O(\log n)$-approximation algorithm for the diamond hitting set problem in the weighted case, which is based on the primal-dual method. We start by defining, in Section \ref{sec:working_LP}, the actual LP relaxation of the problem used by the algorithm, together with its dual. Then, in Section \ref{algo_logn}, we describe the approximation algorithm and, in Section \ref{def_logn_analysis}, we prove that it provides a $O(\log n)$-approximation for the diamond hitting set problem. Finally, in Section \ref{def_logn_gap}, we show that the integrality gap of the natural LP relaxation for the problem (see (\ref{ds_ineq}), page \pageref{ds_ineq}) is $\Theta(\log n)$. This last result is obtained using expander graphs with large girth.

\subsection{The working LP and its dual}
\label{sec:working_LP}

Our approximation algorithm is based on the natural LP relaxation for the diamond hitting set problem, given on page \pageref{ds_ineq}. To simplify the presentation, we do not directly resort to that LP relaxation  but to a possibly weaker relaxation that is constructed during the execution of the algorithm, that we call the {\DEF working LP}. At each iteration, an inequality is added to the working LP. These inequalities, that we name {\DEF blended diamond inequalities}, are all implied by diamond inequalities (\ref{ds_ineq}). The final working LP reads:
\begin{alignat}{4}
\nonumber \textrm{(LP)} \qquad& &\mbox{Min }& &\sum_{v \in V(G)} c_v\,x_v & \\
\nonumber & &\mbox{s.t. } & &\sum_{v \in V(G)} a_{i,v}\,x_v &\geqslant \beta_i & \qquad \forall i \in \{1,\ldots,k\}\\
\nonumber & & & &x_v &\geqslant 0 &\qquad \forall v \in V(G),
\end{alignat}
where $k$ is the total number of iterations of the algorithm. The dual of (LP) is:
\begin{alignat}{4}
\nonumber \textrm{(D)} \qquad& &\mbox{Max }& &\sum_{i = 1}^k \beta_i\,y_i & \\
\nonumber & &\mbox{s.t. } & &\sum_{i = 1}^k a_{i,v}\,y_i &\leqslant c_v  & \qquad \forall v \in V(G)\\
\nonumber & & & &y_i &\geqslant 0 & \qquad \forall i \in \{1,\ldots,k\}.
\end{alignat}

The algorithm is based on the primal-dual method. It maintains a boolean primal solution $x$ and a feasible dual solution $y$. Initially, all variables are set to $0$. Then the algorithm enters its main loop, that ends when $x$ satisfies all diamond inequalities. At the $i$th iteration, a violated inequality $\sum_{v \in V} a_{i,v}\,x_v \geqslant \beta_i$ is added to the working LP and the corresponding dual variable $y_i$ is increased. In order to preserve the feasibility of the dual solution, we stop increasing $y_i$ whenever some dual inequality becomes tight. That is, we stop increasing when $\sum_{j = 1}^i  a_{j,v}\,y_j = c_v$ for some vertex $v$, that is said to be {\DEF tight}. Furthermore, we also stop increasing $y_i$ in case a `collision' occurs (see Section \ref{collisions}). All tight vertices $v$ (if any) are then added to the primal solution. That is, the corresponding variables $x_v$ are increased from $0$ to $1$. The current iteration then ends and we check whether $x$ satisfies all diamond inequalities. If so, then we exit the loop, perform a reverse delete step, and output the current primal solution. 

The precise way the violated blended diamond inequality is chosen is defined in Sections \ref{def_logn_ineqs} and \ref{def_logn_compute_ieq}. It depends among other things on the {\DEF residual cost} (or slack) of the vertices. The residual cost of vertex $v$ at the $i$th iteration is the number $c_v - \sum_{j = 1}^{i-1} a_{j,v}\,y_j$. Note that the residual cost of a vertex is always nonnegative, and zero if and only if the vertex is tight. 

\subsection{The algorithm}
\label{algo_logn}

A formal definition of the algorithm is given in Algorithm~\ref{algo:logn}. 
All the steps are explicit, except those labeled $\star$, which will be specified later.

\begin{algorithm}
\caption{\label{algo:logn}A $O(\log n)$-approximation algorithm for weighted graphs.}
\label{ag:weighted}
\begin{itemize}
\item $X \leftarrow \varnothing$; \quad  $y \leftarrow 0$; \quad $i \leftarrow 0$; \quad  $\mathcal{L} \leftarrow \varnothing$
\item While $X$ is not a hitting set of $G = (V,E)$, repeat the following steps:
\begin{itemize}
\item $i \leftarrow i+1$
\item Let $H$ be the graph obtained by shaving $G - X$
\item Find a reduced graph $\wt{H}$ of $H$
\item Find a diamond $\wt{D}$ in $\wt{H}$ of size $\leqslant\, 6\log_{3/2}|\wt{H}| + 8$ \hfill (using Lemma \ref{lem:reduced_implies_log_diamond})
\item[$\star$] Find an induced subgraph $S$ of $H$, based on $\wt{D}$ \hfill (see Section \ref{def_logn_ineqs})
\item[$\star$] If $S$ is not consistent with $\mathcal{L}$, modify $S$ \hfill (see Section \ref{sec:modify_S})
\item[$\star$] Compute a violated blended diamond inequality $\sum_{v \in V} a_{i,v}\,x_v \geqslant \beta_i$,\\ based on $S$, $\mathcal{L}$ and the residual costs;  add it to (LP) \hfill (see Section \ref{def_logn_compute_ieq})
\item[$\star$] Increase $y_i$ until some vertex becomes tight, or a collision occurs \\ \mbox{}
\hfill (see Section \ref{collisions})
\item[$\star$] Update $\mathcal{L}$ \hfill (see Section \ref{def_logn_update_L_1})
\item[$\star$] Add all tight vertices $v$ to $X$, in a certain order \hfill (see Section \ref{def_logn_order})
\item[$\star$] Re-update $\mathcal{L}$ \hfill (see Section \ref{def_logn_update_L_2})
\end{itemize}
\item $k \leftarrow i$
\item Perform a reverse delete step on $X$ 
\end{itemize}
\end{algorithm}

Above, $\mathcal{L}$ is a collection of triples $(T, B, \{v,w\})$ used to guide the choice of subgraph $S$,
where $T$ and $B$ are internally disjoint $v$--$w$ paths (see Section \ref{sec:modify_S}).

We remark that the set $X$ naturally corresponds to a primal solution $x$, obtained by setting $x_{v}$ to $1$ if $v \in X$, to $0$ otherwise, for every $v\in V(G)$. This vector $x$ satisfies the diamond inequalities (\ref{ds_ineq}) exactly when we exit the \emph{while} loop of the algorithm, that is, when $X$ becomes a hitting set.

The reverse delete step consists in considering the vertices of $X$ in the reverse order in which they were added to $X$ and deleting those vertices $v$ such that $X - \{v\}$ is still a hitting set. Observe that, because of this step, the hitting set $X$ output by the algorithm is inclusion-wise minimal.

The remainder of this section is organized as follows.
In Section \ref{def_logn_ineqs}, we define the `support graph' $S$ of the inequalities and then explain, in Section \ref{sec:modify_S}, how to modify it when it is not consistent with $\mc{L}$. After that, we define the blended diamond inequalities in Section \ref{def_logn_compute_ieq},  we define collisions and explain how to take care of them in Section \ref{collisions}, we specify the insertion order of vertices in the solution in Section \ref{def_logn_order}, and we explain the way list $\mc{L}$ is updated in Sections \ref{def_logn_update_L_1} and \ref{def_logn_update_L_2}.

\subsubsection{The support graph of the inequalities}
\label{def_logn_ineqs}
First, we need some definitions.
Let $H$ be a shaved graph  and let $\wt{H}$ be a reduced graph of $H$ 
(as defined in Section \ref{reductions}). 
Vertices in $V(\wt{H})$ and $V(H) - V(\wt{H})$
are called {\DEF branch vertices} and {\DEF internal vertices} of $H$,
respectively.

We will consider graphs that are equipped with a collection
of specific subgraphs, called `pieces', which have a structure similar
to that of bonds. The first kind of such graphs are simply diamonds:
Consider a diamond $D$, and replace each edge $vw \in E(D)$ with a
path between $v$ and $w$. Each such path is a {\DEF piece} of the resulting
diamond; the vertices $v,w$ are the {\DEF ends} of the piece, the others are the {\DEF internal} vertices of the piece.

A second kind of graph equipped with pieces is a {\DEF necklace}, defined
as any graph obtained as follows. 
First, choose a cycle $C$ (cycles consisting of two parallel edges are allowed).
Then, select a non-empty subset $Z$ of edges of $C$. Next, replace each edge $vw\in E(C) - Z$ of $C$
with a path between $v$ and $w$. Finally, replace each edge $vw\in Z$ 
either by two internally disjoint $v$--$w$ paths, or 
by a cactus whose block-graph is a path, in such a way that $v$ and $w$ lie in distinct endblocks of the cactus and in no other block. In both cases, the subgraph by which an edge 
of $C$ has been replaced is called a {\DEF piece}
of the necklace; ends and internal vertices of the piece are defined as expected.
See Figure~\ref{fig:mutation} for an illustration.

A piece $Q$ of a diamond or necklace is {\DEF simple} if $Q$ is isomorphic to a path, 
{\DEF double} otherwise.
Let us point out that, while pieces and bonds look very similar at first sight,
there are some differences between the two notions. (Notice for instance that
a piece of a necklace could consist of a single edge or a cycle.)

A diamond or necklace $S$ together with its pieces is {\DEF rooted} in $H$ if 
\begin{itemize}
\item $S$ is an induced subgraph of $H$; 
\item every end of a piece of $S$ is a branch vertex of $H$;
\item every internal vertex of a piece of $S$ is an internal vertex of $H$, and
\item there is no edge in $H$ between an internal vertex of a piece of $S$ and a vertex in $V(H) - V(S)$.
\end{itemize}
(Observe that this definition also depends on $\wt{H}$, since the latter graph determines
which vertices of $H$ are branch vertices.)

\begin{figure}[hbt]
\centering \includegraphics[scale=0.9]{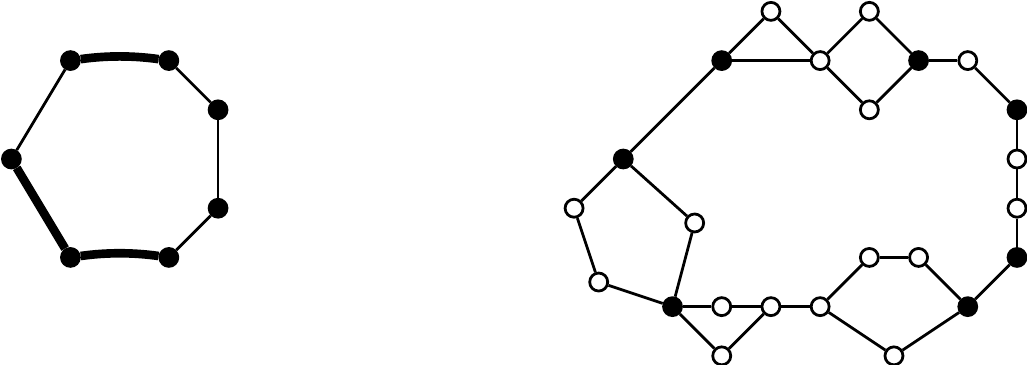}
\caption{A cycle $C$, where edges in $Z$ are thicker (left) and a necklace built from $C$ (right).}
\label{fig:mutation}
\end{figure}

\begin{lemma}
\label{lem:diamond-types}
Let $H$ be a shaved graph.
Let $\wt H$ be a reduced graph of $H$. 
Then, given a diamond $\wt D \subseteq \wt H$,
one can find in polynomial time a subgraph $S$ of $H$
such that one of the following conditions is satisfied:
\begin{enumerate}[(i)]
\item \label{DT-diamond}
$S$ is a diamond rooted in $H$, with at most $||\wt D||$ pieces;
\item \label{DT-cycle}
$S$ is a necklace rooted in $H$, with at most $||\wt D||$ pieces;
\item \label{DT-parallel}
$S$ consists of two vertices with at least four parallel edges between them;
\item \label{DT-K4}
$S$ is isomorphic to $K_4$.
\end{enumerate}
\end{lemma}
\begin{proof}
First, we associate to each edge  
of $\wt{H}$ a corresponding {\DEF primitive subgraph} in $H$, defined as follows. 
Consider an edge $e\in E(\wt{H})$. If $e$ was already present in $H$, then its primitive subgraph is the edge itself and its two ends. 
Otherwise, the primitive subgraph of $e$ is the bond whose reduction produced $e$.
In particular, if $e$ has a twin edge $e'$, then the primitive subgraphs of $e$ and $e'$ coincide. 
The primitive subgraph $J$ of a subgraph $\wt{J} \subseteq \wt{H}$ is defined simply as the union of the primitive subgraphs of every edge in $E(\wt{J})$. Note that primitive subgraphs are well-defined, thanks to Lemma \ref{lem:disjoint}. Also, notice that the primitive subgraph of a subgraph of $\wt{H}$ is not defined \emph{per se}, but with respect to the bond reductions which produced $\wt{H}$ from $H$. 

Let $K$ denote the subgraph of $\wt{H}$ induced by $V(\wt{D})$. Consider an induced subgraph $K'$ of $K$ that contains a diamond and is vertex-minimal with that property, that is, $K' - v$ is a forest of cacti for every $v \in V(K')$.
Let $\mu$ be the maximum number of parallel edges between pairs of adjacent vertices in $K'$.

First, suppose $\mu=1$. As the reader will easily check, the minimality of $K'$ implies that either $K'$ is a simple diamond, or $K'$ is isomorphic to $K_{4}$. 

In the first case ($K'$ is a simple diamond), the primitive subgraph $S$ of $K'$ in $H$ 
can be seen as a diamond with $||K'||$ pieces. It follows from the definition of primitive subgraphs and the fact that $H$ is shaved, that $S$ is rooted in $H$. Also, we trivially have $||K'|| \leqslant ||\wt D||$; thus, $S$ has at most $||\wt D||$ pieces. Hence, $S$ satisfies (\ref{DT-diamond}).

In the second case ($K'\simeq K_{4}$), we may assume that the primitive subgraph of $K'$ is not isomorphic to $K'$ (otherwise, (\ref{DT-K4}) holds). Then, there is an edge $e\in E(K')$ such that the primitive subgraph of $e$ is a path of length at least $2$. Let then $S$ be the primitive subgraph of $K' - e$. The subgraph $S$ is induced in $H$ and is a diamond with $||K' - e||$ pieces. Since, similarly as before, $S$ is rooted in $H$ and $||K' - e|| \leqslant || \wt D||$, it follows that $S$ satisfies (\ref{DT-diamond}).

Next, assume $\mu=2$. Let $v, w$ be two vertices of $K'$ that are connected by two parallel edges $e_{1}, e_{2}$ in $K'$. Let $P$ be a shortest $v$--$w$ path in the graph $K' - \{e_{1}, e_{2}\}$. (Observe that such a path exists, since $K'$ contains a spanning diamond.) Then, all vertices of $K'$ are included in $P$. Thus, $K'$ is a simple cycle with some (and at least one) of its edges replaced by pairs of parallel edges. Let $S$ be the primitive subgraph of $K'$ in $H$. Then, $S$ is a necklace rooted in $H$ and has at most $|K'| \leqslant ||K'|| \leqslant ||\wt D||$ pieces. Hence, $S$ satisfies (\ref{DT-cycle}).

Now, suppose $\mu \geqslant 3$. Then, $K'$ consists of two vertices $v$ and $w$ connected by $\mu$ parallel edges. If $H$ contains a pair of vertices with at least four parallel edges between them, then (\ref{DT-parallel}) holds for the obvious choice of $S$, and we are done. Thus, assume there are at most three edges between any two vertices in $H$. 

Let $S$ be the primitive subgraph of $K'$. If $K'$ contains no pair of twin edges, then $S$ consists of at most three edges between $v$ and $w$ and at least $\mu - 3$ longer paths between $v$ and $w$. By removing vertices from $S$ if necessary, we can ensure that $S$ satisfies (\ref{DT-diamond}). Otherwise, $K'$ contains at least one pair of twin edges. In this case, by removing vertices from $S$ if necessary, we can ensure that $S$ satisfies (\ref{DT-cycle}), or (\ref{DT-diamond}) if $S$ has three parallel edges between $v$ and $w$

Finally, we note that the above proof can be turned without difficulty into a polynomial-time algorithm computing $S$, thus concluding the proof.
\end{proof}

We say that an induced subgraph $S \subseteq H$ is of {\DEF type 1} ({\DEF type 2, 3, 4}, resp.) if it satisfies condition $(i)$  (condition $(ii)$, $(iii)$, $(iv)$, resp.) of the above lemma.

\subsubsection{Modifying the graph $S$}
\label{sec:modify_S}

Consider some iteration $i$ of the algorithm, the corresponding shaved graph $H$
and the induced subgraph $S$ of $H$ produced at that iteration. Here, we explain how to modify $S$ when it is not consistent with $\mathcal{L}$. If $S$ is of type 3 or 4, there is no need to modify it. Hence, we assume that $S$ is of type 1 or 2 for the rest of this section. First, we need to introduce more terminology. 

Consider a piece $Q$ of $S$ containing a cycle $C$. Then $C$ is a block of $Q$. A vertex $v$ of $C$ is said to be an {\DEF end} of the cycle $C$ if $v$ is an end of the piece $Q$ or $v$ belongs to a block of $Q$ distinct from $C$. Observe that $C$ has always two distinct ends. The cycle $C$ has also two {\DEF handles}, defined as the two $v$--$w$ paths in $C$, where $v$ and $w$ are the two ends of $C$. A handle is {\DEF trivial} if it has no internal vertex, {\DEF non-trivial} otherwise.

The two handles of $C$ are labelled {\DEF top} and {\DEF bottom} as follows. Suppose $C$ has two non-trivial handles, and compute the minimum residual cost of an internal vertex in each handle. If this minimum is achieved in exactly one handle, then this handle is the top handle. If, on the other hand, the minimum is achieved in both, then the tie is broken arbitrarily, unless the cycle $C$ was considered in a previous iteration. In this case, we ensure that the tie is always broken in the same way (actually, we can use the list $\mathcal{L}$ to determine this, see below). Now, if $C$ has only one non-trivial handle, then it is defined to be the top handle. Finally, if both handles of $C$ are trivial (that is, $C$ is a cycle of length $2$), then one of them is chosen arbitrarily and called the top handle. In each of these three cases, the bottom handle is the other handle, as expected.

Now, we may give a precise definition of $\mathcal{L}$: it is a collection of triples $(T,B,\{v,w\})$ satisfying \emph{all} of the following conditions:
\begin{itemize}
\item $v, w$ are two distinct vertices of $H$, 
\item $T$ and $B$ are two $v$--$w$ paths in $H$ that are internally disjoint, 
\item $T$ has at least one internal vertex, 
\item internal vertices of $T$ and $B$ have exactly two neighbors in $H$.
\end{itemize}
Moreover, we require that 
\begin{itemize}
\item for every two distinct triples $(T,B,\{v,w\}), (T',B',\{v',w'\})$ in  $\mathcal{L}$, no two of the four paths $T, T', B, B'$ have an internal vertex in common.
\end{itemize}

The graph $S$ is {\DEF consistent with} a triple $(T,B,\{v,w\})$ in $\mathcal{L}$ if \emph{any} of the following conditions is satisfied:
\begin{enumerate}[(C1)] 
\item \label{C-disj} $S$ contains no internal vertex of any of the paths $T$, $B$,
\item \label{C-top} $S$ contains $T$ in one of its simple pieces and no internal vertex of $B$,
\item \label{C-cycle-1} $S$ is of type 1 and contains the cycle $T \cup B$,
\item \label{C-cycle-2} $S$ is of type 2 and contains $T \cup B$ in one of its (double) pieces.
\end{enumerate}
We say that $S$ is consistent with the collection $\mathcal{L}$ if $S$ is consistent with every triple in $\mathcal{L}$.

Next, we explain how to modify $S$ if it is not consistent with $\mathcal{L}$.
Roughly speaking, the purpose of this step is to ensure that no vertex in a bottom handle becomes tight before some vertex in the corresponding top handle becomes tight. This property is crucial for our analysis of the algorithm.

We iteratively modify $S$, until it becomes consistent with $\mathcal{L}$: Let $(T,B,\{v,w\})\in \mathcal{L}$ be such that $S$ is not consistent with $(T,B,\{v,w\})$. (Recall that $S$ is of type 1 or 2, by assumption.) Because (C\ref{C-disj}) is not satisfied, $S$ contains an internal vertex of $T$ or $B$. 

\begin{claim}
\label{cl:inside}
If $P$ is any path in $\{T,B\}$ such that one of its internal vertices is contained in $S$, then the whole path $P$ is contained in $S$.
\end{claim}

\begin{proof}
Let $u$ be an internal vertex of $P$ contained in $S$.  Because $u$ is an internal vertex of $P$, it has exactly two neighbors in $H$. Because $u$ is a vertex of $S$ and $S$ is a diamond or a necklace, $u$ has at least two neighbors in $S$. Thus, the two neighbors of $u$ in $H$ are in $S$. By repeating this argument, it follows that the whole path $P$ is contained in $S$.
\end{proof}

By what precedes, we may assume that some path $P \in \{T,B\}$ has at least one internal vertex in $S$. By Claim \ref{cl:inside}, $P$ is entirely contained in $S$. Because every end of a piece of $S$ is a branch vertex, all vertices of $P$ are in the same piece of $S$, say $Q$. Let $P'$ denote the other path in $\{T,B\}$.\medskip

\begin{claim}
\label{cl:special_necklace}
Assume that the cycle $P \cup P' = T \cup B$ is entirely contained in $S$. Then, $S$ is a necklace with two pieces, one piece is a cycle with ends $v$ and $w$, and the other piece is either a $v$--$w$ path or also a cycle with ends $v$ and $w$.
\end{claim}

\begin{proof}
Because $S$ is not consistent with $(T,B,\{v,w\})$, neither (C\ref{C-cycle-1}) nor (C\ref{C-cycle-2}) is satisfied. It follows that $S$ is a necklace, and the $v$--$w$ paths $P$ and $P'$ are contained in different pieces of $S$. In particular, $S$ has two distinct pieces sharing two distinct vertices, namely $v$ and $w$. This implies that $S$ has exactly two pieces, and $v$ and $w$ are the ends of both pieces. One of the pieces contains $P$ (namely, $Q$) and the other contains $P'$. The rest of the claim follows easily.
\end{proof}

\noindent \emph{Case 1.} No internal vertex of $P'$ is contained in $S$.

First, assume that the path $P'$ has no internal vertex. Then, the whole cycle $P \cup P'$ is contained in $S$ because $S$ is induced in $H$. By Claim \ref{cl:special_necklace}, we can transform $S$ into a diamond, by deleting a (possibly empty) subset of the vertices of $S$.

Second, assume that $P'$ has at least one internal vertex. Because $v$ and $w$ are vertices of $S$ that have at least one neighbor outside $S$, they are branch vertices. Thus, $v$ and $w$ are the ends of the piece $Q$. 

If $Q$ is simple, then we have $P = B$ and $P' = T$ because otherwise $S$ would be consistent with $(T,B,\{v,w\})$. In this case, we redefine $S$ as the subgraph of $H$ induced by $(V(S) - V(B)) \cup V(T)$. (Thus, we ``replace'' $B$ with $T$ in $S$.) 

If $Q$ is double, then $Q$ is a cycle with ends $v$ and $w$, and $S$ is a necklace. Then, we transform $S$ into a diamond, by redefining $S$ as $Q \cup P'$.\medskip

\noindent \emph{Case 2.} Some internal vertex of $P'$ is contained in $S$. 

By Claim \ref{cl:inside}, it follows that all vertices of $P'$ are contained in $S$. Hence, $S$ contains the whole cycle $P \cup P'$. Again, by Claim \ref{cl:special_necklace}, we can transform $S$ into a diamond, by deleting a subset of the vertices of $S$.\medskip

In all the cases above, we either transform $S$ into a diamond, or make $S$ consistent with $(T,B,\{v,w\})$ without creating a new inconsistency. It follows that the modification process is finite, and produces a new induced subgraph $S$ still satisfying the requirements of Lemma \ref{lem:diamond-types}, that is moreover consistent with all triples in $\mathcal{L}$. Clearly, modifying $S$ can be done in polynomial time.

\subsubsection{Computing the inequality}
\label{def_logn_compute_ieq}

Let $C$ be a cycle contained in a piece $Q$ of $S$. If at least one handle of $C$ is non-trivial,
then we denote by $\trc(C)$ (resp.\ $\brc(C)$) the minimum residual cost of an internal vertex in the top handle (resp.\ bottom handle) of $C$. The convention is that $\brc(C) = \infty$ if the bottom handle
is trivial.
It will be convenient to say that a vertex {\DEF belongs} to a handle if it is an internal vertex
of that handle.

To the graph $S$ we associate a unique {\DEF blended diamond inequality} of the form
\begin{equation}
\label{eq:p_d_ineq}
\sum_{v \in V(G)} a_{i,v}\,x_v \geqslant 1,
\end{equation}
whose support is contained in the vertex set of $S$. For convenience, we call $S$ the {\DEF support graph}\footnote{This is an abuse of notation, since the support of the inequality is not always equal to $S$.} of the inequality. If $S$ is of type 1, 3, or 4, we let
$$
a_{i,v} := \left\{ 
\begin{array}{ll}
1 &\textrm{if }v \in V(S),\\
0 &\textrm{if }v \in V(G) - V(S).
\end{array}\right.
$$
If $S$ is of type 2, we choose a cycle of $S$ and declare it to be {\DEF special}, as follows:
If there is a cycle $C$ in a piece of $S$ such that $\trc(C)=\brc(C)$, then we select such a cycle.
Otherwise, we select an arbitrary cycle contained in a piece of $S$.
Then, we let
$$
a_{i,v} := \left\{
\begin {array}{ll}
   1 &\mbox {if $v \in V(S)$ and $v$ belongs to no handle}, \\
   0
   &\mbox {if $v$ belongs to the top handle of a non-special cycle $C$ and $\trc(C) < \brc(C)$}, \\
   1
   &\mbox {if $v$ belongs to the bottom handle of a non-special cycle $C$ and $\trc(C) < \brc(C)$}, \\
   \frac{1}{2} &\mbox {if $v$ belongs to a handle of a non-special cycle $C$ and $\trc(C) = \brc(C)$}, \\
   1 &\mbox {if $v$ belongs to a handle of the special cycle}, \\      
   0  &\mbox {if } v \in V(G) - V(S). \\
\end{array} \right.
$$

\begin{lemma}
Every blended diamond inequality~\eqref{eq:p_d_ineq} is implied by the diamond and non-negativity inequalities.
\end{lemma}

\begin{proof}
If $S$ is of type 1, 3, or 4, Inequality~\eqref{eq:p_d_ineq} is a diamond inequality. Otherwise, \eqref{eq:p_d_ineq} is easily seen to be a convex combination of two diamond inequalities. 
\end{proof}

\subsubsection{Taking care of collisions}
\label{collisions}

When increasing the variable $y_{i}$, the residual cost 
$c_v - \sum_{j = 1}^i  a_{j,v}\,y_j$ of every vertex $v$ in $S$ decreases, at a speed given
by the coefficient $a_{i,v}$. 
Also, for every cycle $C$ included in a piece of $S$, we have that $\trc(C)$ and $\brc(C)$ decrease,
possibly at different speeds. We could simply increase $y_{i}$
until some vertex $v$ becomes tight (that is, until its residual cost drops to $0$).
However, by doing so, it could be that $\trc(C) \leqslant \brc(C)$ no longer holds for some piece $Q$ and
some cycle $C$ in $Q$. (For instance, this would happen if $\brc(C)$ decreases much faster than $\trc(C)$.)
We will need that $\trc(C) \leqslant \brc(C)$ remains true in future iterations, so that
the top and bottom handles of $C$ do not interchange 
(this is used in the proof of Lemma~\ref{lem:heart_of_log(n)}).
For this reason, we have to keep track of `collisions', as we now explain.

A {\DEF collision} occurs if, for some cycle $C$ in a piece of $S$, we had $\trc(C) < \brc(C)$
at the beginning of the iteration, and $\trc(C)$ and $\brc(C)$ become equal
when increasing $y_{i}$. As mentioned in the algorithm, we stop
increasing $y_{i}$ when a collision occurs. 
If no vertex became tight during the current iteration, then the algorithm
will simply keep the same graph $S$ in the next iteration, but will change the coefficient 
of every vertex $v$ belonging to a handle of $C$. The new coefficients
are equal to $1/2$ (or to $1$ if $C$ becomes the special cycle); 
in particular, $\trc(C)$ and $\brc(C)$ will decrease at the same speed in the future.

Finally, let us point out that the number of consecutive iterations executed by the algorithm
before some vertex becomes tight is bounded by the number of cycles contained in pieces of $S$, since
a cycle can be involved in at most one collision. Hence, the total number
of iterations of the algorithm is polynomial.

\subsubsection{Updating the list $\mathcal{L}$ (first pass)}
\label{def_logn_update_L_1}

We update the collection $\mathcal{L}$ as follows.  
If $S$ is of type 1, 3, or 4, we leave $\mathcal{L}$ unchanged.
If, on the other hand, it is of type 2, we add to $\mathcal{L}$ all triples $(T,B,\{v,w\})$ that were not yet present in $\mathcal{L}$, such that $T$ is the top handle, $B$ is the bottom handle and $\{v,w\}$ are the ends of some cycle contained in a piece of $S$, and $T$ has at least one internal vertex. Because $S$ is consistent with the original list $\mathcal{L}$, the updated list $\mathcal{L}$ satisfies the required properties.

\subsubsection{The insertion order}
\label{def_logn_order}

When several vertices are added to $X$ in the same iteration, we pick an enumeration of the triples of $\mathcal{L}$, say $(T_1,B_1,\{v_1,w_1\})$, \ldots, $(T_\ell,B_\ell,\{v_\ell,w_\ell\})$ having the property that all triples $(T_i,B_i,\{v_i,w_i\})$ such that $X$ contains an internal vertex of both $T_i$ and $B_i$ come first, but otherwise arbitrarily. Then, we insert last all vertices $u$ such that $u$ is an internal vertex of $T_1$, followed by all vertices $u$ such that $u$ is an internal vertex of $B_1$, followed by all vertices $u$ such that $u$ is an internal vertex of $T_2$, and so on, ending with all vertices $u$ such that $u$ is an internal vertex of $B_\ell$.

\subsubsection{Updating the list $\mathcal{L}$ (second pass)}
\label{def_logn_update_L_2}

After vertices have been added to $X$ (hence, deleted from $H$), we remove from $\mathcal{L}$ all triples $(T,B,\{v,w\})$ such that $V(T)$ or $V(B)$ has a non-empty intersection with the new solution $X$. We also remove all triples $(T,B,\{v,w\})$ such that one of the vertices of $V(T)$ or $V(B)$ will be removed when shaving $G-X$.

\subsection{Analysis of the algorithm}
\label{def_logn_analysis}

Before proceeding with the analysis of the algorithm, we need a lemma.

\begin{lemma}
\label{lem:simple_charging}
Consider a $2$-connected graph, some of whose vertices and edges are marked. If no marked vertex is incident to a marked edge, then the total number of marked vertices and edges is at most the total number of edges.
\end{lemma} 

\begin{proof}
Consider any feasible assignment of marks to vertices and edges of a $2$-connected graph. First, suppose that all vertices are marked and, therefore, that no edge is marked. Since the graph is $2$-connected, the number of vertices is at most the number of edges, thus the result follows. 

Now, suppose that there exists an unmarked vertex. If no vertex of the graph is marked, we are done. Otherwise, by connectivity, there exists an edge $uv$ such that $u$ is unmarked and $v$ is marked, hence $uv$ is unmarked. Unmark $v$ and mark $uv$. This operation does not change the total number of marked elements and, by applying it iteratively as long as there exists a marked vertex, we eventually obtain a graph without any marked vertex. The result follows.
\end{proof}

\begin{lemma}
\label{lem:heart_of_log(n)}
Let $X$ be the hitting set output by the algorithm. Then
\begin{equation}
\label{eq:heart_of_log(n)}
\sum_{v \in X} a_{i,v} \leqslant \left(12 \log_{3/2} n + 16\right)\beta_i
\end{equation}
for all $i \in \{1,\ldots,k\}$.
\end{lemma}

\begin{proof}
First, we recall that $\beta_i = 1$ for all the inequalities in the working LP relaxation.

Consider the $i$th iteration of the algorithm.
In what follows, $H$ stands for the graph $H$ at the $i$th iteration, 
and $S$ is the support graph of the $i$th inequality of the working LP.

If $S$ is {\em not} of type 2, then the left hand side of \eqref{eq:heart_of_log(n)} is simply the number of vertices of $X$ contained in $S$.
It follows that, if $S$ is of type 3, then the left hand side is at most $2$; and, if $S$ is of type 4, then the left hand side is at most $4$.

Next, we consider the case where the type of $S$ is 1 or 2.

\begin{claim}
\label{claim:pieceX}
Let $Q$ be a piece of $S$. Exactly one of the following four cases occurs:
\begin{enumerate}[(a)]
\item\label{enum:none} $X$ contains no internal vertex of $Q$,
\item\label{enum:cut} $X$ contains exactly one vertex of $Q$, and this vertex is a cutvertex of $Q$,
\item\label{enum:opposite} $X$ contains exactly two vertices of $Q$, and they belong to opposite handles of a cycle of $Q$,
\item\label{enum:lots} $X$ contains exactly one vertex per cycle of $Q$, each belonging to some handle of the corresponding cycle.
\end{enumerate}
\end{claim}
\begin{proof}
Let $Z$ be the set of vertices that $X$ contained at the beginning of the $i$th iteration.
(Thus, $Z$ is the set of vertices that became tight at some iteration $j$ with $j<i$.)
Every vertex $u \in V(Q) \cap X$ has a corresponding {\DEF witness}, namely,
a diamond $D_{u} \subseteq G$ such that $V(D_u) \cap (X \cup Z) = \{u\}$.
(Such a subgraph exists because $u$ was kept during the reverse delete step.)
This witness $D_{u}$ cannot contain any vertex that was removed during the shaving of $G - Z$ at
the beginning of iteration $i$ since, by definition, these vertices are not included
in any diamond of $G - Z$. Hence, the diamond $D_{u}$ is a subgraph of $H - (X - \{u\})$.

Suppose $X$ contains some internal vertex $u$ of $Q$ (otherwise, 
\eqref{enum:none} trivially holds).
Let $w_{1}$ and $w_{2}$ be the two ends of $Q$.
Since no internal vertex of $Q$ has a neighbor in $V(H) - V(Q)$ in $H$, the diamond
$D_{u}$ must contain $w_{1}$, $w_{2}$, and every cutvertex of $Q$. 
Thus, $X$ contains none of $w_{1}$ and $w_{2}$.

If $u$ is a cutvertex of $Q$, then there cannot be another internal vertex $v$ of $Q$ in $X$, for otherwise $D_{v}$ contains $u$. 
Hence, \eqref{enum:cut} holds in this case.

Now, assume that $u$ is not a cutvertex of $Q$. By the previous observation, we may also assume
that no other vertex in $X$ is a cutvertex of $Q$.
This implies that each vertex in $X \cap V(Q) - \{w_{1},w_{2}\}$ belongs
to a handle of a cycle in $Q$. Let $A_{1}, A_{2}$ be the two handles
of the cycle in $Q$ containing $u$, with $u\in V(A_{1})$. 

If $A_{2} \not\subseteq D_{u}$,
then $X$ must contain some internal vertex $v$ of $A_{2}$: Otherwise, replacing
the path $A_{1}$ with $A_{2}$ in $D_{u}$ gives a diamond in $G-X$, a contradiction.
It follows that $u$ and $v$ are the only internal vertices of $Q$ included $X$, since $\{u, v\}$
separates $w_{1}$ from $w_{2}$ in $Q$. Hence, \eqref{enum:opposite} holds.

Finally, suppose $A_{2} \subseteq D_{u}$. In this case, $X$ contains no internal vertex of $A_{2}$. Moreover, $D_{u}$ contains exactly one handle of each cycle in $Q$ that is distinct from $A_{1} \cup A_{2}$.
Consider such a cycle $C$. The set $X$ must contain some internal vertex of the handle of $C$ that is not in $D_{u}$ (since otherwise we could again find a diamond in $G-X$ using that handle). Furthermore, $X$ contains at most two internal vertices of $C$. Otherwise, either $X$ contains two vertices $v$, $v'$ belonging to opposite handles of $C$, and $\{v,v'\}$ would separate $u$ from $w_1$ or $w_2$; or $X$ contains two vertices $v$, $v'$ belonging to the same handle of $C$, and $D_{v}$ would contain $v'$. Therefore, \eqref{enum:lots} holds.
\end{proof}

If $S$ is of type 1, then by Claim~\ref{claim:pieceX} we can see the vertices in $X\cap V(S)$ as marks
on the pieces of $S$ such that, if the interior of a piece is marked, then none of
its ends are. Hence, by Lemma \ref{lem:simple_charging}, 
$|X \cap V(S)|$ is at most the number of pieces in $S$, which in turn is at most
$6 \log_{3/2} n + 8$.

From now on, we assume that $S$ is of type 2. We split the left hand side of \eqref{eq:heart_of_log(n)} into two parts: the vertices that are internal vertices of some piece of $S$, and the branch vertices. 

\begin{claim}
\label{claim:piece_contrib}
Consider a piece $Q$ of $S$. The contribution of the internal vertices of $Q$ to the left hand side of \eqref{eq:heart_of_log(n)} is at most $1$ if $Q$ does not contain the special cycle of $S$, and at most $2$ if $Q$ contains the special cycle of $S$.
\end{claim}
\begin{proof}
By our choice of coefficients for inequality \eqref{eq:heart_of_log(n)}, the claim holds in cases (\ref{enum:none}), (\ref{enum:cut}) and (\ref{enum:opposite}) of Claim~\ref{claim:pieceX}. Thus, we assume that case (\ref{enum:lots}) holds. First, we show that $X$ contains exactly one vertex belonging to the top handle of each cycle contained in $Q$.

We use the following property: For all $(T,B,\{v,w\})$ in $\mathcal{L}$, the minimum residual cost of a vertex belonging $T$ is less than or equal to the minimum residual cost of a vertex belonging to $B$. This holds when $(T,B,\{v,w\})$ is added to $\mathcal{L}$ (by the definition of top and bottom handles). The inequality is maintained as long as the support graph $S$ is unchanged (see Section \ref{collisions}). Moreover, the inequality is also maintained when a new support graph $S$ is chosen, because we ensure that $S$ is always consistent with $\mathcal{L}$ (see Section \ref{sec:modify_S}). Consequently, the inequality is maintained through all subsequent iterations.

Suppose, by contradiction, that there is a cycle $C$ inside the piece $Q$, with top handle $T$ and bottom handle $B$, such that $X$ contains a vertex belonging to $B$, say $u$. Because case (\ref{enum:lots}) holds, (the final set) $X$ contains no vertex belonging to $T$. Because of the property above, at the iteration in which $u$ was added to (the evolving set) $X$, at least one vertex belonging to $T$ was also added to $X$, say $u'$. Now, because of our particular insertion order, it must be the case that $u$ was added after $u'$ in $X$. Since $u$ survived the reverse delete step, there is a diamond $D_{u}$ in $G$ witnessing the fact that $u$ is in (the final set) $X$. This diamond $D_{u}$ does not contain $u'$. Hence, $D_{u}$ can be easily transformed into a diamond $D_{u'}$ containing $u'$ and disjoint from $X$, a contradiction. It follows that $X$ contains exactly one vertex belonging to the top handle of each cycle in $Q$.

Now, we may assume that $Q$ contains at least three cycles, and at least one cycle $C$ such that $\trc(C) = \brc(C)$, because otherwise the claim trivially holds. Consider such a cycle $C$, with top handle $T$, bottom handle $B$, and ends $v$, $w$. In every subsequent iteration such that $(T,B,\{v,w\})$ survives in $\mathcal{L}$, case (C\ref{C-disj}) or (C\ref{C-cycle-2}) arises (with respect to the corresponding subgraph $S$). It follows that $\trc(C) = \brc(C)$ holds in every subsequent iteration. In particular, at the iteration in which a vertex of $T$ is added to $X$, a vertex of $B$ is also added to $X$. 
Therefore, $Q$ contains at most one cycle $C$ such that $\trc(C) = \brc(C)$. This is due to the fact that case (\ref{enum:lots}) arises and to our insertion order.

If $Q$ does not contain the special cycle, then the contribution of the internal vertices of $Q$ to the left hand side of \eqref{eq:heart_of_log(n)} is clearly at most $1/2 \leqslant 1$. If $Q$ contains the special cycle, then the contribution of the internal vertices of $Q$ to the left hand side of \eqref{eq:heart_of_log(n)} is at most $1$. (Recall that the special cycle is chosen among the cycles $C$ contained in a piece of $S$ and such that $\trc(C) = \brc(C)$, if any.)
\end{proof}

The coefficients in \eqref{eq:heart_of_log(n)} corresponding to the branch vertices of $S$ are all equal to $1$. Hence, combining Claim
\ref{claim:piece_contrib} and Lemma~\ref{lem:simple_charging}, we deduce that the left hand side of \eqref{eq:heart_of_log(n)} is at most $12 \log_{3/2} n + 16$ (since $S$ has at most $6 \log_{3/2} n + 8$ pieces).
\end{proof}

\begin{theorem}
\label{Ologn}
Algorithm~\ref{algo:logn} is a $O(\log n)$-approximation for the diamond hitting set problem.
\end{theorem}

\begin{proof}
Letting $\alpha = 12 \log_{3/2} n + 16$, we have
$$
\sum_{v \in X} c_v \; = \; \sum_{v \in X} \Big(\sum_{i = 1}^k a_{i,v}\,y_i\Big) \; = \; \sum_{i = 1}^k \Big(\sum_{v \in X} a_{i,v}\Big)\,y_i \; \leqslant \; \sum_{i=1}^k \alpha\,\beta_i\,y_i \; \leqslant \; \alpha\,OPT,
$$
where the first equality holds because all vertices in $X$ are tight, and the first inequality follows from Lemma \ref{lem:heart_of_log(n)}. The result follows.
\end{proof}

\subsection{Integrality gap}
\label{def_logn_gap}

\begin{proposition}
The integrality gap of the LP defined by non-negativity and diamond inequalities is $\Theta(\log n)$.
\end{proposition}

\begin{proof}
By Theorem \ref{Ologn}, we know that the integrality gap is $O(\log n)$. 
Now, we show that the integrality gap is also $\Omega(\log n)$, using expander graphs with large girth.

Let us first recall some standard notions from the theory of expanders. Let $G$ be a $d$-regular graph with $|G| \geqslant 2$. 
The {\DEF spectral gap} of $G$ is $d - \lambda(G)$, where $\lambda(G)$ denotes the second largest eigenvalue of the adjacency matrix of $G$.
The {\DEF vertex-expansion} $h_v(G)$ is defined as the minimum of
$$
\frac{|N(S) - S|}{|S|}
$$
over all subsets $S\subset V(G)$ with $|S| \leqslant \frac{|G|}{2}$. (Here, $N(S)$ denotes the set of vertices of $G$ having a neighbor in $S$.)
It is well-known that the spectral gap of $G$ can be used to derive a lower-bound on its vertex-expansion: 
\begin{equation}
\label{eq-CB}
h_v(G) \geqslant \frac{1}{2d}\left(d - \lambda(G)\right);
\end{equation}
see for instance the survey by Hoory, Linial, and Wigderson~\cite[p. 454]{HLW06}.

Lubotzky, Phillips, and Sarnak~\cite{LPS88} proved that, for $d=6$ and infinitely many values of $n$, there exists a $d$-regular graph $G$ on $n$ vertices
with $\lambda\left(G\right) \leqslant 2\sqrt{d-1}$ and girth at least $\frac{4}{3}\log_{d-1}n$ (see also Margulis~\cite{M88}, and Biggs and Boshier~\cite{BB90}). 
By~\eqref{eq-CB}, the vertex-expansion of $G$ satisfies
$$
h_v\left(G\right) \geqslant \frac{1}{2d}\left(d - 2\sqrt{d-1}\right) = \frac{1}{12}\left(6 - 2\sqrt{5}\right) > 0.
$$

It is also known that the treewidth\footnote{See Diestel~\cite{D05} for a definition.} 
$\tw(H)$ of a graph $H$ satisfies
$$
\tw(H) \geqslant h_v(H) \cdot \frac{n}{4} - 1,
$$
see for instance Grohe and Marx~\cite [Proposition~1]{GM09}.  It follows that 
$$
\tw\left(G\right) \geqslant \xi n - 1,
$$
where $\xi:=\frac{1}{48}\left(6 - 2\sqrt{5}\right)$.
Since removing a vertex from a graph decreases its treewidth by at most 1 and forest of cacti have treewidth at most 2, this implies that the minimum size of a hitting set satisfies
$$
OPT \geqslant\xi n - 3.
$$
On the other hand, the minimum size of a diamond in $G$ is at least the girth, that is, at least $\frac{4}{3}\log_{5}n$. Thus, setting 
$$
x^*_v:= \frac{3}{4\log_{5}n}
$$
yields a feasible solution of the linear relaxation. The value of the objective function for $x^*$ is $\frac{3n}{4\log_{5}n}$. Therefore, the integrality gap of the LP is at least
$$
4\log_{5}n \cdot \frac{OPT}{3n} \geqslant4\log_{5}n \cdot \frac{\xi n - 3}{3n} = \Omega(\log n).
$$
\end{proof}

% ======================================================================

% ======================================================================
\section{A 9-approximation algorithm}
\label{const_weighted}

In this section, we give a primal-dual 9-approximation algorithm for the diamond hitting set problem. We start with a description of the algorithm in Section \ref{algo}. This algorithm makes use of the sparsity inequalities. In order to describe them, we first bound the number of edges in a forest of cacti in Section \ref{bounding}; using this bound, in Sections \ref{basic} and \ref{lifted} we introduce the sparsity inequalities, prove their validity, and show that they satisfy a key inequality that we need in the analysis of the algorithm. Finally, in Section \ref{analysis}, we prove that our algorithm provides a 9-approximation for the diamond hitting set problem.

In the whole section, $\mxc$ is a global parameter of our algorithm 
which is set to $5$. 
(We remark that the analysis below could be adapted to other values of $\mxc$,
but this would not give an approximation factor better than $9$.)

% ======================================================================
\subsection{The algorithm}
\label{algo}

Our 9-approximation algorithm for the diamond hitting set problem is very similar to the $O(\log n)$-approximation algorithm. The main difference is that we use a different set of inequalities to build the working LP relaxation. (The working LP relaxation and its dual are defined in Section \ref{sec:working_LP}, on page \pageref{sec:working_LP}.) See Algorithm~\ref{algo:constant} for
a description of the algorithm.

\begin{algorithm}
\caption{\label{algo:constant}A 9-approximation algorithm.}
\begin{itemize}
\item $X \leftarrow \varnothing$; \quad $y \leftarrow 0$; \quad $i \leftarrow 0$; \quad $\mathcal{L} \leftarrow \varnothing$
\item While $X$ is not a hitting set of $G=(V,E)$, repeat the following steps:
\begin{itemize}
\item $i \leftarrow i+1$
\item Let $H$ be the graph obtained by shaving $G-X$
\item Find a reduced graph $\wt{H}$ of $H$
\item If $\wt{H}$ contains a diamond $\wt{D}$ with at most $2\mxc-1$ edges, then let $\sum_{v \in V} a_{i,v}\,x_v \geqslant \beta_i$ be a blended diamond inequality deduced from $\wt{D}$ as in Section \ref{def_logn_compute_ieq}
\item Otherwise, in $\wt{H}$, no two cycles of size at most $\mxc$ share an edge. In this case, let $\sum_{v \in V} a_{i,v}\,x_v \geqslant \beta_i$ be the extended sparsity inequality with support $V(H)$
\item Check the inequality $\sum_{v \in V} a_{i,v}\,x_v \geqslant \beta_i$ w.r.t.\ $\mathcal{L}$, modify it if necessary, add it to the working LP
\item Increase $y_i$ until some vertex becomes tight, or a collision occurs
\item Update $\mathcal{L}$
\item Add all tight vertices to $X$, in a certain order
\item Re-update $\mathcal{L}$
\end{itemize}
\item $k \leftarrow i$
\item Perform a reverse delete step on $X$
\end{itemize}
\end{algorithm}

The following sections are devoted to the definition of the extended sparsity inequalities. In the next paragraph we comment the last five steps in the main loop.

First, we only modify the current inequality $\sum_{v \in V} a_{i,v}\,x_v \geqslant \beta_i$ if it is a blended diamond inequality. This is done as described in Section \ref{sec:modify_S}. Note that none of the modifications increases the number of pieces of the support graph $S$. 

Second, the tracking of collisions is performed as before (see Section \ref{collisions}). 

Third, $\mathcal{L}$ is updated similarly as before (see Section \ref{def_logn_update_L_1}): if the support graph of the current inequality contains double pieces (in case the current inequality is an extended sparsity inequality, the support graph is the whole graph $H$), we add all triples $(T,B,\{v,w\})$ that were not yet present in $\mathcal{L}$, where $T$ is the top handle, $B$ is the bottom handle, and $v$, $w$ are the ends of a cycle contained in a double piece of the support graph. 

Finally, the insertion order and the second update of $\mathcal{L}$ are done exactly as previously (see Sections \ref{def_logn_order} and
\ref{def_logn_update_L_2}).
% ======================================================================
\subsection{Bounding the number of edges in a forest of cacti}
\label{bounding}

The following lemma provides a bound on the number of edges in a forest of cacti. For $i \in \{2, \dots, \mxc\}$, we denote by $\gamma_i(G)$ the number of cycles of length $i$ of a graph $G$.

\begin{lemma}
\label{lem:cactus_ub}
Let $F$ be a forest of cacti with $k$ components and let $\mxc \geqslant 2$. Then
$$
||F|| \leqslant \frac{\mxc+1}{\mxc} (|F| - k) + \sum_{i=2}^{\mxc} \frac{\mxc-i+1}{\mxc}\,\gamma_i(F). 
$$
\end{lemma}

\begin{proof}
Denote by $\gamma_{> \mxc}(F)$ the number of cycles of $F$ whose length exceeds $\mxc$. We have
\begin{equation}
\label{eq:cactus_eq}
||F|| = |F| - k + \sum_{i=2}^{\mxc} \gamma_i(F) + \gamma_{> \mxc}(F).
\end{equation}
In the right hand side, the first two terms represent the number of edges in a spanning forest of $F$, while the last terms give the number of edges that should be added to obtain the forest of cacti $F$.

Because every two cycles in $F$ are edge disjoint, we have
$$
||F|| \geqslant \sum_{i=2}^{\mxc} i \, \gamma_i(F) + (\mxc+1) \, \gamma_{> \mxc}(F).
$$
Combining this with \eqref{eq:cactus_eq}, we get
\begin{equation}
\label{eq:cyc_ub}
\gamma_{> \mxc}(F) \leqslant \frac{1}{\mxc} \Big (|F| - k - \sum_{i=2}^{\mxc} (i-1) \gamma_i(F) \Big).
\end{equation}
From (\ref{eq:cactus_eq}) and (\ref{eq:cyc_ub}), we finally infer
$$
||F|| \leqslant |F| - k +  \sum_{i=2}^{\mxc} \gamma_i(F) + \frac{1}{\mxc} \Big (|F| - k - \sum_{i=2}^{\mxc} (i-1)\,\gamma_i(F) \Big) \leqslant \frac{\mxc+1}{\mxc} (|F| - k) + \sum_{i=2}^{\mxc} \frac{\mxc-i+1}{\mxc}\,\gamma_i(F). 
$$
\end{proof}

% ======================================================================
\subsection{The sparsity inequalities}
\label{basic}

We define the {\DEF load} of a vertex $v$ in a graph $G$ as 
$$ 
\ell_G(v) := d_G(v) - \sum_{i=2}^{\mxc} \lambda_i \, \gamma_i(G,v),
$$
where, for $i \in \{2, \ldots, \mxc\}$, $\gamma_i(G,v)$ denotes the number of cycles of length $i$ incident to $v$ in $G$ and 
$$
\lambda_i := \frac{\mxc-i+1}{\lfloor i/2 \rfloor\,\mxc}.
$$
For $\mxc = 5$, we have
$$
\lambda_2 = \frac{4}{5}\ , \qquad \lambda_3 = \frac{3}{5}\ , \qquad \lambda_4 = \frac{1}{5}\ , \quad \textrm{and} \quad \lambda_5 = \frac{1}{10}.
$$

\begin{lemma}
\label{lem:bsc}
Let $X$ be a hitting set of a graph $G$ where no two cycles of length at most $\mxc$ share an edge. Then,
\begin{equation}
\label{bsc}
\sum_{v \in X} \Big( \ell_G(v) - \frac{\mxc+1}{\mxc} \Big) \geqslant  ||G||  - \frac{\mxc+1}{\mxc}|G| -  \sum_{i = 2}^{\mxc} \frac{\mxc-i+1}{\mxc}\,\gamma_i(G) + \frac{\mxc+1}{\mxc}.
\end{equation}
\end{lemma}

We call Inequality~\eqref{bsc} a {\DEF sparsity inequality}.

\begin{proof}[Proof of Lemma~\ref{lem:bsc}]
For $i \in \{2,\ldots,\mxc\}$ and $j \in \{0,\ldots,i\}$, we denote by $\xi_i^j$ the number of cycles of $G$ that have length $i$ and exactly $j$ vertices in $X$. 

Letting $||X||$ and $|\delta(X)|$ respectively denote the number of edges of $G$ with both ends in $X$ and the number of edges of $G$ having an end in $X$ and the other in $V(G) - X$, we have 
\begin{eqnarray*}
\sum_{v \in X} \ell_G(v) & =&2 ||X|| + |\delta(X)| - \sum_{i=2}^{\mxc} \sum_{j=1}^i j \, \lambda_i \, \xi_i^j\\ 
& = & ||X|| + ||G|| - ||G - X|| -  \sum_{i=2}^{\mxc} \sum_{j=1}^i j\,\lambda_i\, \xi_i^j\\
& \geqslant & ||X|| + ||G|| - \frac{\mxc+1}{\mxc}(|G-X|-1) - \sum_{i=2}^{\mxc} \frac{\mxc-i+1}{\mxc}\, \xi_i^0 - \sum_{i=2}^{\mxc} \sum_{j=1}^i j\,\lambda_i\,\xi_i^j,
\end{eqnarray*}
where the last inequality follows from Lemma \ref{lem:cactus_ub} applied to the forest of cacti $G-X$ (notice that $\gamma_i(G-X) = \xi_i^0$).

Because no two cycles of length at most $\mxc$ share an edge and, in a cycle of length $i$, each subset of size $j$ induces a subgraph that contains at least $2j-i$ edges, we have 
$$
||X|| \geqslant  \sum_{i=2}^{\mxc} \sum_{j=1+\lfloor\frac{i}{2}\rfloor}^{i} (2j-i)\,\xi_i^j.
$$

Thus, we obtain
$$
\sum_{v \in X} \ell_G(v) \geqslant \sum_{i=2}^{\mxc} \sum_{j=1+\lfloor\frac{i}{2} \rfloor}^{i} (2j-i)\,\xi_i^j + ||G|| - \frac{\mxc+1}{\mxc}(|G-X|-1) - \sum_{i=2}^{\mxc} \frac{\mxc-i+1}{\mxc}\,\xi_i^0 - \sum_{i=2}^{\mxc} \sum_{j=1}^i j\,\lambda_i\, \xi_i^j.
$$
We leave it to the reader to check that, in the right hand side of the last inequality, the total coefficient of $\xi_i^j$ is at least $- \frac{\mxc-i+1}{\mxc}$, for all $i \in \{2,\ldots,\mxc\}$ and $j \in \{0,\ldots,i\}$. Hence,
$$ 
\sum_{v \in X} \ell_G(v)
\geqslant ||G|| - \frac{\mxc+1}{\mxc} |G| + \frac{\mxc+1}{\mxc} |X| + \frac{\mxc+1}{\mxc}  - \sum_{i = 2}^{\mxc} \frac{\mxc-i+1}{\mxc} \, \gamma_i(G).
$$
Inequality~\eqref{bsc} follows.
\end{proof}

% ======================================================================
\subsection{The extended sparsity inequalities}
\label{lifted}

Consider a shaved graph $H$ and denote by $\wt{H}$ a reduced graph of $H$. Suppose that, in $\wt{H}$, no two cycles of length at most $\mxc$ share an edge. Thus, the graph $H$ can be uniquely decomposed into simple or double pieces corresponding to edges or pairs of parallel edges in $\wt{H}$. Here, the pieces of $H$ are defined as follows: let $v$ and $w$ be two adjacent vertices of $\wt{H}$, and let $\wt{J}$ denote the subgraph of $\wt{H}$ induced by $\{v,w\}$. The primitive subgraph of $\wt{J}$ in $H$, say $J$, is a {\DEF piece} of $H$ with ends $v$ and $w$. The vertices of $H$ are of two types: the branch vertices are those that survive in $\wt{H}$, and the other vertices are internal to some piece of $H$.

Consider a double piece $Q$ of $H$ (if any) and a cycle $C$ contained in $Q$. As before, denote by $\trc(C)$ (resp.\ $\brc(C)$) the minimum residual cost of a vertex in the top handle (resp.\ bottom handle) of $C$. Also, we choose a cycle of $Q$ (if any) and declare it to be special. If possible, the special cycle is chosen among the cycles $C$ contained in $Q$ with $\trc(C) = \brc(C)$. (So \emph{every} double piece of $H$ has a special cycle.)

The {\DEF extended sparsity inequality} for $H$ reads
\begin{equation}
\label{eq:lifted}
\sum_{v \in V(H)} a_{v}\,x_v \geqslant \beta,
\end{equation}
where
$$
a_v := \left\{
\begin {array}{ll}
\displaystyle \ell_{\wt{H}}(v) - \frac{\mxc+1}{\mxc}  &\mbox {if $v$ is a branch vertex}, \\
1 &\mbox {if $v$ is an internal vertex of a simple piece}, \\
2 &\mbox {if $v$ is an internal vertex of a double piece and does not belong to any handle}, \\
0
&\mbox {if $v$ belongs to the top handle of a cycle $C$ with $\trc(C) < \brc(C)$}, \\
2
&\mbox {if $v$ belongs to the bottom handle of a cycle $C$ with $\trc(C) < \brc(C)$}, \\
1 &\mbox {if $v$ belongs to a handle of a cycle $C$ with $\trc(C) = \brc(C)$, or $C$ is special},
\end{array} \right.
$$
and
$$
\beta := ||\wt{H}||  - \frac{\mxc+1}{\mxc}|\wt{H}| -  \sum_{i=2}^{\mxc} \frac{\mxc-i+1}{\mxc} \gamma_i(\wt{H}) + \frac{\mxc+1}{\mxc}.
$$
%In the definition of $a_v$ above, we always assume that the cycle $C$ is contained in a double piece of $H$. 
By convention, the support graph of the extended sparsity inequality \eqref{eq:lifted} is defined to be $H$.

\begin{lemma}
Let $H$ be a graph and let $\wt{H}$ be a reduced graph of $H$ such that no two cycles of length at most 
$\mxc$ share an edge. Then Inequality~\eqref{eq:lifted} is valid, that is, 
$$
\sum_{v \in X} a_v \geqslant \beta
$$
whenever $X$ is a hitting set of $H$. 
\end{lemma}

\begin{proof}
Let $Y$ denote the set of branch vertices that are included in $X$. So $Y = V(\wt{H}) \cap X$. Then $Y$ is not necessarily a hitting set of $\wt{H}$. Indeed, $\wt{H} - Y$ is a forest of cacti $F$ to which a certain number of {\DEF extra} edges are added: the extra edges are those corresponding to the vertices of $X$ that are internal to some piece, that is, vertices of $X - Y$. By our choice of coefficients, 
$$
\sum_{v \in X - Y} a_v \ge ||\wt{H}-Y|| - ||F||.
$$
In other words, the left hand side is at least the number of extra edges. 

The rest of the proof closely follows the proof of Lemma \ref{lem:bsc}. This time, all computations are made within the reduced graph. Precisely, $\xi_i^j$ denotes the number of cycles of $\wt{H}$ that have length $i$ and exactly $j$ vertices in $Y$, $||Y||$ denotes the number of edges of $\wt{H}$ with both ends in $Y$ and $|\delta(Y)|$ denotes the number of edges of $\wt{H}$ having one end in $Y$ and the other in $V(\wt{H})-Y$. Then, we have
\begin{eqnarray*}
\sum_{v \in Y} \ell_{\wt{H}}(v) + \sum_{v \in X - Y} a_v& =&2 ||Y|| + |\delta(Y)| - \sum_{i=2}^{\mxc} \sum_{j=1}^i j\,\lambda_i\,\xi_i^j + 
||\wt{H}-Y|| - ||F||\\ 
& = & ||\wt{H}|| + ||Y|| - ||F|| -  \sum_{i=2}^{\mxc} \sum_{j=1}^i j\,\lambda_i\,\xi_i^j.
\end{eqnarray*}
By using Lemma \ref{lem:cactus_ub} to bound $||F||$ and the inequality $\xi_i^0 \geqslant \gamma_i(F)$, we obtain
\begin{eqnarray*}
\sum_{v \in Y} \ell_{\wt{H}}(v) + \sum_{v \in X - Y} a_v
& \geqslant & ||\wt{H}|| + ||Y|| - \frac{\mxc+1}{\mxc}(|F|-1) - \sum_{i=2}^{\mxc} \frac{\mxc-i+1}{\mxc}\, \gamma_i(F) - \sum_{i=2}^{\mxc} \sum_{j=1}^i j\,\lambda_i\,\xi_i^j\\
& \geqslant & ||\wt{H}|| + ||Y|| - \frac{\mxc+1}{\mxc}(|\wt{H}|-|Y|-1) - \sum_{i=2}^{\mxc} \frac{\mxc-i+1}{\mxc}\, \xi_i^0 - \sum_{i=2}^{\mxc} \sum_{j=1}^i j\,\lambda_i\,\xi_i^j.
\end{eqnarray*}

After performing the same steps as in the proof of Lemma \ref{lem:bsc}, in the graph $\wt{H}$ this time, we find
$$
\sum_{v \in Y} \ell_{\wt{H}}(v) + \sum_{v \in X - Y} a_v \geqslant ||\wt{H}||  - \frac{\mxc+1}{\mxc}|\wt{H}| + \frac{\mxc+1}{\mxc} |Y| + \frac{\mxc+1}{\mxc} -  \sum_{i=2}^{\mxc} \frac{\mxc-i+1}{\mxc} \gamma_i(\wt{H}).
$$
This proves the validity of the extended sparsity Inequality~\eqref{eq:lifted}.
\end{proof}

% ======================================================================
\subsection{Analysis of the algorithm}
\label{analysis}

We now prove a key inequality that will be used in the analysis of the algorithm. 

\begin{lemma}
\label{sparsity_bound}
Let $H$ be a shaved graph and let $\wt{H}$ be a reduced graph of $H$. Suppose that, in $\wt{H}$, 
every diamond has at least $2q=10$ edges. Then,
$$
\sum_{v \in X} a_{v} \leqslant  8 \, \beta
$$
for every minimal hitting set $X$ of $H$.
\end{lemma}
%%%% OLD STATEMENT
%
%\begin{lemma}
%\label{sparsity_bound}
%Let $H$ be a shaved graph and $\wt{H}$ a reduced graph of $H$. Suppose that, in $\wt{H}$, no two cycles of length at most $\mxc$ share an edge. Then,
%$$
%\sum_{v \in X} a_{v} \leqslant  8 \, \beta
%$$
%%
%for every minimal hitting set $X$ of $H$.
%\end{lemma}

\begin{proof}
Let $Y := V(\wt{H}) \cap X$. We note that
no two cycles of length at most $5$ in $\wt{H}$ have an edge in common,
since $\wt{H}$ has no diamond with at most $9$ edges.

We have to prove that 
\begin{equation}
\label{ineq_analysis}
\sum_{v \in Y} \Big(  \ell_{\wt{H}}(v) - \frac{\mxc+1}{\mxc}\Big)+ \sum_{v \in X - Y} a_v 
\leqslant 8 \cdot \Big(||\wt{H}|| - \frac{\mxc+1}{\mxc}|\wt{H}|  -  \sum_{i=2}^{\mxc} \frac{\mxc-i+1}{\mxc} \gamma_i(\wt{H}) + \frac{\mxc+1}{\mxc}\Big).
\end{equation}
We claim that
$$
\sum_{v \in X - Y} a_v = ||\wt{H}-Y|| - ||F||,
$$
where the sum is taken over all vertices that are included in $X$ and that are internal to some piece. 
Indeed, since $X$ is minimal, there are four ways in which $X$ can intersect internal vertices of a piece $Q$ (see Claim \ref{claim:pieceX} above):
\begin{itemize}
\item $X$ contains no internal vertex of $Q$,
\item $X$ contains exactly one vertex of $Q$, and this vertex is a cutvertex of $Q$,
\item $X$ contains exactly two vertices of $Q$, and they belong to opposite handles of a cycle of $Q$,
\item $X$ contains exactly one vertex per cycle of $Q$, each belonging to some handle of the corresponding cycle.
\end{itemize}
In the three first cases, the choice for the coefficients ensures that the equality holds. In the fourth case, using the same arguments as in Claim \ref{claim:piece_contrib} above, we prove that $X$ contains one vertex in the top handle of each cycle of $Q$.
The left hand side can be rewritten as:
\begin{eqnarray*}
\sum_{v \in Y} \Big( \ell_{\wt{H}}(v) - \frac{\mxc+1}{\mxc}\Big)+  \sum_{v \in X - Y} a_v  &= & 2 ||Y|| + |\delta(Y)| - \sum_{i=2}^{\mxc} \sum_{j=1}^i j\,\lambda_i \,\xi_i^j - \frac{\mxc+1}{\mxc}|Y|\\
&&+ ||\wt{H}-Y|| - ||F|| \\[2ex]
&\leqslant& 2||\wt{H}|| -  |\delta(Y)| - \sum_{i=2}^{\mxc} \sum_{j=1}^i j\,\lambda_i \,\xi_i^j - \frac{\mxc+1}{\mxc}|Y|\\
&& - 2 ||F|| - (||\wt{H} - Y|| - ||F||). 
\end{eqnarray*}
Let $k$ denote the number of components of the forest of cacti $F$. Since $||F|| = |F| - k + \sum_{i=2}^{\mxc} \gamma_i(F) + \gamma_{> \mxc}(F)$ and $|F| = |\wt{H}| - |Y|$, we have
\begin{eqnarray*}
\sum_{v \in Y} \Big( \ell_{\wt{H}}(v) - \frac{\mxc+1}{\mxc}\Big)+  \sum_{v \in X - Y} a_v &= & 2 ||\wt{H}|| - |\delta(Y)| - \sum_{i=2}^{\mxc} \sum_{j=1}^i j\,\lambda_i \,\xi_i^j - \frac{\mxc+1}{\mxc}|Y|\\ 
& & - 2 \Big(|\wt{H}| - |Y| - k + \sum_{i=2}^{\mxc} \gamma_i(F) + \gamma_{> \mxc}(F) \Big) \\ 
& & - (||\wt{H} - Y|| - ||E(F)||)\\ 
&\leqslant & 2 ||\wt{H}|| - 2 |\wt{H}| - \sum_{i=2}^{\mxc} \sum_{j=1}^i j\,\lambda_i \, \xi_i^j - 2 \sum_{i=2}^{\mxc} \gamma_i(F) - 2\, \gamma_{> \mxc}(F) \\
& & - |\delta(Y)| - (||\wt{H} - Y|| - ||F||) + \frac{\mxc-1}{\mxc}|Y| + 2 k. 
\end{eqnarray*}
Now, let $\nu_i$ ($i=\{2,\dots, \mxc\}$) be defined as 
$\nu_i := 0$ for $i \neq 3$ and $\nu_3 := \frac{3}{5}$.
Notice that $\nu_i \leqslant j \lambda_i$ for all $i \in \{2,\ldots,\mxc\}$ and $j \in \{1,\ldots,i\}$.
Also, $\nu_i \leqslant \frac{\mxc-i+1}{\lfloor i/2 \rfloor\,\mxc}$ 
for all $i \in \{2,\ldots,\mxc\}$. We rewrite the last inequality as 
\begin{eqnarray*}
\sum_{v \in Y} \Big( \ell_{\wt{H}}(v) - \frac{\mxc+1}{\mxc}\Big)+  \sum_{v \in X - Y} a_v &\leqslant & 2 ||\wt{H}|| - |\wt{H}| - \sum_{i=2}^{\mxc} \sum_{j=1}^i j\,\lambda_i \,\xi_i^j -2 \sum_{i=2}^{\mxc} \gamma_i(F) - 2\, \gamma_{> \mxc}(F) \\
& &- \sum_{i=2}^{\mxc} \nu_i (\xi_i^0 - \gamma_i(F)) + \sum_{i=2}^{\mxc} \nu_i (\xi_i^0 - \gamma_i(F)) \\
&& - |\delta(Y)| - (||\wt{H} - Y|| - ||F||) + \frac{\mxc-1}{\mxc}|Y| + 2 k
\\
&\leqslant & 2 ||\wt{H}|| - |\wt{H}| - \sum_{i=2}^{\mxc} \nu_i\,\gamma_i(\wt{H}) - \sum_{i=2}^{\mxc} \sum_{j=1}^i (j\,\lambda_i - \nu_i) \,\xi_i^j\\
& &  - (2 - \nu_i) \sum_{i=2}^{\mxc} \gamma_i(F) - 2\, \gamma_{> \mxc}(F) + \sum_{i=2}^{\mxc} \nu_i (\xi_i^0 - \gamma_i(F)) \\ 
&& - |\delta(Y)| - (||\wt{H} - Y|| - ||F||) + \frac{\mxc-1}{\mxc}|Y| + 2 k\\
&\leqslant & 2 ||\wt{H}|| - |\wt{H}|  - \sum_{i=2}^{\mxc} \nu_i\,\gamma_i(\wt{H}) + \sum_{i=2}^{\mxc} \nu_i (\xi_i^0 - \gamma_i(F))\\ 
&&  - |\delta(Y)| - (||\wt{H} - Y|| - ||F||) + \frac{\mxc-1}{\mxc}|Y| + 2 k. 
\end{eqnarray*}
Therefore, to prove Inequality~\eqref{ineq_analysis}, it suffices to show the following:
\begin{eqnarray*}
6 ||\wt{H}|| + \Big(2- 8 \cdot \frac{\mxc+1}{\mxc}\Big) |\wt{H}| - \sum_{i = 2}^{\mxc} \Big( 8 \cdot \frac{\mxc-i+1}{\mxc} - \nu_i\Big)\,\gamma_i(\wt{H}) + 8\cdot\frac{\mxc+1}{\mxc} \\
+ |\delta(Y)| - \frac{\mxc-1}{\mxc}|Y| + ||\wt{H} - Y|| - ||F|| -2k  - \sum_{i=2}^{\mxc} \nu_i (\xi_i^0 - \gamma_i(F))
 & \geqslant& 0.
 \end{eqnarray*}
We actually prove a slightly stronger inequality:
\begin{eqnarray*}
6 ||\wt{H}|| + \Big(2- 8 \cdot \frac{\mxc+1}{\mxc}\Big) |\wt{H}| - \sum_{i = 2}^{\mxc} \Big( 8 \cdot \frac{\mxc-i+1}{\mxc} - \nu_i\Big)\,\gamma_i(\wt{H}) \\
+ |\delta(Y)| - \frac{\mxc-1}{\mxc}|Y| + ||\wt{H} - Y|| - ||F|| -2k  - \sum_{i=2}^{\mxc} \nu_i (\xi_i^0 - \gamma_i(F))
 & \geqslant& 0.
 \end{eqnarray*}
We proceed in two steps. We first show that 
\begin{equation}
\label{part1}
      6 ||\wt{H}|| + \Big(2- 8 \cdot \frac{\mxc+1}{\mxc}\Big) |\wt{H}| - \sum_{i = 2}^{\mxc} \Big( 8 \cdot \frac{\mxc-i+1}{\mxc} - \nu_i\Big)\,\gamma_i(\wt{H}) \geqslant 0,
\end{equation}
then we show that
\begin{equation}
\label{part2}
	|\delta(Y)| + ||\wt{H} - Y|| - ||F|| \geqslant \frac{\mxc-1}{\mxc}|Y| + 2k + \sum_{i=2}^{\mxc} \nu_i (\xi_i^0 - \gamma_i(F)).
\end{equation}

\medskip
Let us start with Inequality~(\ref{part1}). By multiplying by $\frac{1}{3|\wt{H}|}$, we rewrite it as:
$$
\frac{1}{|\wt{H}|} \cdot \left (2 ||\wt{H}||  - \frac{1}{3} \cdot \sum_{i = 2}^{\mxc} \Big( 8 \cdot \frac{\mxc-i+1}{i\,\mxc} - \frac{\nu_i}{i}\Big)\,i\,\gamma_i(\wt{H}) \right) \geqslant \frac{8 \cdot \frac{\mxc+1}{\mxc} - 2}{3} . 
$$
The above inequality has the following simple interpretation: the average load of a vertex should be at least the right hand side, where the load of vertex $v$ is redefined as
$$
\ell'_{\wt{H}}(v) := d_{\wt{H}}(v) - \frac{1}{3} \cdot \sum_{i = 2}^{\mxc} \Big(8 \cdot \frac{\mxc-i+1}{i\,\mxc} - \frac{\nu_i}{i}\Big)\,\gamma_i(\wt{H}, v).
$$
Recalling that $\mxc = 5$, $\nu_i = 0$ for $i \neq 3$ and $\nu_3 = \frac{3}{5}$, we get that the right hand side is equal to $\frac{38}{15}$ and the load is as follows:
\begin{eqnarray*}
\ell'_{\wt{H}}(v) &= &d_{\wt{H}}(v) 
- \frac{1}{3} \cdot \Big(8 \cdot \frac{5-2+1}{2 \cdot 5} \Big)\,\gamma_2(\wt{H}, v)
- \frac{1}{3} \cdot \Big(8 \cdot \frac{5-3+1}{3 \cdot 5} - \frac{1}{5}\Big)\,\gamma_3(\wt{H}, v)\\
&&- \frac{1}{3} \cdot \Big(8 \cdot \frac{5-4+1}{4 \cdot 5} \Big)\,\gamma_4(\wt{H}, v)
- \frac{1}{3} \cdot \Big(8 \cdot \frac{5-2+1}{5 \cdot 5} \Big)\,\gamma_5(\wt{H}, v)\\
&= &d_{\wt{H}}(v) 
- \frac{16}{15}\,\gamma_2(\wt{H}, v)
- \frac{7}{15}\,\gamma_3(\wt{H}, v)
- \frac{4}{15}\,\gamma_4(\wt{H}, v)
- \frac{8}{75}\,\gamma_5(\wt{H}, v)
\end{eqnarray*}
As $\wt{H}$ is reduced, the minimum load of a vertex is $3-\frac{7}{15} = \frac{38}{15}$ (think of a vertex incident to three edges, two of which determine a cycle of length $3$). 
As a consequence, the average load of a vertex is at least the right hand side, 
and Inequality~\eqref{part1} follows.

\medskip
Let us now prove Inequality~(\ref{part2}). Since $\mxc=5$, $\nu_i = 0$ for $i \neq 3$ and $\nu_3 = \frac{3}{5}$, the inequality reads $|\delta(Y)| + ||\wt{H}-Y|| - ||F|| \geqslant \frac{4}{5} |Y| + 2k + \frac{3}{5} (\xi_3^0 - \gamma_3(F))$. We will prove the following slightly stronger inequality:
\begin{equation} \label{toprove}
	|\delta(Y)| + ||\wt{H}-Y|| - ||F|| \geqslant  |Y| + 2k + \frac{3}{5} (\xi_3^0 - \gamma_3(F)).
\end{equation}
To prove it, we use arguments that are similar to those used in \cite{cghw98}. First observe that the quantities $||\wt{H}-Y|| - ||F||$ and $\xi_3^0 - \gamma_3(F)$ respectively correspond to the number of extra edges and the number of extra triangles contained in $\wt{H} - Y$. Furthermore, each extra edge can generate at most one extra triangle, since $\wt{H}$ has no diamond with at most $9$ edges.

We build a bipartite graph $J$ as follows. Start with $\wt{H}$, and contract each of the $k$ components of $\wt{H}-Y$ into a single vertex (as usual, we keep the newly created parallel edges, if any,
and we remove the loops). Then, remove all edges having both endpoints in $Y$. 

It follows from the fact that $X$ is a minimal hitting set of $H$ that, for each vertex $v \in Y$, there exists a diamond $D_v$ in $\wt{H}$ that is vertex-disjoint from $Y-\{v\}$ and edge-disjoint from $E(\wt{H}-Y) - E(F)$. Moreover, for each extra edge $e \in E(\wt{H}-Y) - E(F)$, there exists a diamond $D_e$ in $\wt{H}$ that is vertex-disjoint from $Y$ and edge-disjoint from $E(\wt{H}-Y) - E(F) - \{e\}$. We call these diamonds a {\DEF witness} for $v$ and $e$, respectively. 

In particular, for every $v\in Y$, we can choose a component $K$ of $\wt{H} - Y$ such that there are at least two edges between $v$ and $K$ in $J$; we call the pair $(v, K)$ a {\DEF primary pair}. We remove from $J$ one edge between $v$ and $K$ for each primary pair $(v, K)$.
Noticed that we removed exactly $|Y|$ edges from $J$; thus, 
to prove Inequality~(\ref{toprove}), it is enough to show that
\begin{equation}
\label{eq:intermediate_claim}
||J|| + ||\wt{H}-Y|| - ||F|| \geqslant 2k + \frac{3}{5} (\xi_3^0 - \gamma_3(F)).
\end{equation}

In the remainder of this proof, our aim is to show that, for each component $K$ of $\wt{H}-Y$, the number of edges of $J$ incident to $K$ plus the number of extra edges in $K$ is at least $2$ plus $3/5$ times the number of extra triangles in $K$, that is:
\begin{equation}
\label{eq:final_claim}
d_J(K) + ||K|| - ||K \cap F|| \geqslant 2 + \frac{3}{5}(\gamma_3(K) - \gamma_3(K \cap F)).
\end{equation}
Clearly, Inequality~\eqref{eq:final_claim} implies Inequality~\eqref{eq:intermediate_claim}. 

First, we recall that, since $\tilde{H}$ is reduced, 
every vertex of $\tilde{H}$ has at least three neighbors.
Consider a component $K$ of $\wt{H}-Y$. Let $\eta = \eta(K) := ||K|| - ||K \cap F||$ denote the number of extra edges in $K$, and $\tau = \tau(K) := \gamma_3(K) - \gamma_3(K \cap F)$ denote the number of extra triangles in $K$. 
With these notations, Inequality~\eqref{eq:final_claim} becomes
\begin{equation}
\label{eq:final_claim_restated}
d_J(K) + \eta \geqslant 2 + \frac{3}{5}\tau.
\end{equation}

Note that $\eta \geqslant \tau$, since no two extra triangles have an edge in common.
Therefore, we may assume that one of the two following cases occurs 
(otherwise \eqref{eq:final_claim_restated} holds):
\begin{enumerate}[(i)]
\item $d_J(K) = 0$ and $\eta \leqslant 4$, or
\item $d_J(K) = 1$ and $\eta \leqslant 2$. 
\end{enumerate} 

Both cases can be handled using similar arguments, hence we treat them
in parallel and only highlight the differences when necessary.

In case (i), $K$ is a component of $\wt{H}$, and hence every vertex of $K$
has at least three distinct neighbors in $K$. In case (ii), 
there are vertices in $V(K)$ having a neighbor in $Y$ in the graph $\wt{H}$;
these vertices are said to be {\DEF special}. 
Note that $K$ belongs to at most one primary pair $(v, K)$, because otherwise
we would have $d_{J}(K) \geqslant 2$. It follows that there are at most two special vertices
in $K$. Moreover, in $K$, every non-special vertex has at least three distinct neighbors,
while every special vertex has at least two of them. 

Let $L := K \cap F$. Each extra edge in $K$ has a witness in $\wt{H}$, and this witness 
avoids all the other extra edges and the vertices in $Y$, thus $L$ is a component of $F$, and hence
$L$ is a cactus.

First, we prove that every cycle in $L$ has at most $7$ vertices. 
Arguing by contradiction, assume $C$ is a cycle of $L$ with $|C| \geqslant 8$.
This cycle is a block of $L$, and each vertex of $C$ either has degree $2$ in $L$, 
or is a cutvertex of $L$. Let $t$ be the number of vertices of $C$ having degree $2$ in $L$.
Now, for each cutvertex $v$ of $L$ included in $C$, there is at least one
endblock $B_{v}$ of $L$ such that $v$ separates $B_{v}$ from $C$ in $L$.  
Observe that either $|B_{v}| = 2$, and thus one of its vertices has only one neighbor in $L$,
or $|B_{v}| \geqslant 3$  and $B_{v}$ is a simple cycle, and 
thus at least two of its vertices have degree $2$ in $L$.
This directly gives the following lower bound on $\eta$: 
$$
\eta \geqslant
\left\{
\begin{array}{lcl}
t/2 + (|C| - t)  &   &  \text{in case (i)}, \\
t/2 + (|C| - t) - 1 &   &  \text{in case (ii)}.
\end{array}
\right.
$$
(The $-1$ in case (ii) comes from the existence of special vertices in $K$.)
In case (i), since $\eta \leqslant 4$, we must have $\eta=4$, $|C|=t=8$, and $L=C$.
However, by taking the union of $C$ with an extra edge of $K$ we obtain a diamond in $K$ having $9$ edges, 
a contradiction.
In case (ii), we directly get a contradiction, because $\eta \leqslant 2$
and $t/2 + (|C| - t) - 1 \geqslant 3$. (The last inequality is derived using
$t \leqslant |C|$ and $|C| \geqslant 8$.)

Consider an extra triangle in $K$. If there is only one extra edge $e$ in that triangle,
then the unique cycle $C$ in $D_{e} - e$ is a cycle of $L$ which has at least one
edge in common with the triangle.
(Recall that $D_{e}$ denotes the witness of $e$ in $K$.) 
As we have seen, we must have $|C| \leqslant 7$. However, this implies 
$||D_{e}|| \leqslant 9$, contradicting the fact that every diamond in $\wt{H}$
has at least $10$ edges.
Thus, every extra triangle contains at least two extra edges. 
Since the extra triangles are edge disjoint, this implies
\begin{equation}
\label{eq:eta_tau}
\eta \geqslant 2\tau.
\end{equation}

In case (ii), Inequality~\eqref{eq:eta_tau} allows us to easily 
show that~\eqref{eq:final_claim_restated} holds: First, suppose $\tau \geqslant 1$.
Then, by~\eqref{eq:eta_tau}, we have 
$$
\eta \geqslant 2\tau \geqslant 1 + \frac{3}{5}\tau,
$$
as desired. Now, assume $\tau = 0$. Here, we only need to show $\eta \geqslant 1$.
Suppose the contrary, that is, $\eta = 0$. Then $K = L$, and hence $K$ is a cactus.
Recall that $K$ has at most two special vertices, and that each of these vertices has
at least two neighbors in $K$. Thus, if some endblock $B$ of $K$ is isomorphic
to $K_{2}$ or to a cycle consisting of two parallel edges, then some vertex of $B$
has only one neighbor in $K$, a contradiction. This implies that 
all the endblocks of $K$ are simple cycles.
Using this observation, it is easily seen that $\eta \geqslant 1$ 
if there are more than one endblock in $K$. 
Thus, $K$ consists of a single block, and hence $K$ is a simple cycle. However, since $|K| \geqslant 3$, 
there is a non-special vertex in $K$ with degree $2$, a contradiction.
Therefore, we must have $\eta \geqslant 1$, which concludes case (ii).

Now, consider case (i).
Here, it is easily seen that $\eta \geqslant 2$: First, $|L| = |K| \geqslant 4$, since
every vertex of $K$ has at least three distinct neighbors in $K$. Thus, if $L$ has only one endblock,
then $L$ must be a simple cycle on at least $4$ vertices, implying $\eta \geqslant  2$. Similarly,
if $L$ has at least two endblocks, then by considering 
two such endblocks we deduce again $\eta \geqslant 2$.
This implies that Inequality \eqref{eq:final_claim_restated} holds if $\tau=0$. Hence,
we may assume $\tau \geqslant 1$. 

If $\eta \geqslant 3$, then \eqref{eq:eta_tau} gives
$$
2 + \frac{3}{5}\tau \leqslant 2 + \frac{3}{10}\eta 
= \eta + 2 - \frac{7}{10}\eta 
= \eta - \frac{1}{10} 
\leqslant \eta.
$$
Thus, it remains to handle the case where
$\eta=2$ and $\tau=1$. We will show that this case cannot happen, because it
leads to a contradiction.

Let $e_{1}, e_{2}, f$ denote the three edges forming the unique extra triangle in $K$,
with $e_{1}$ and $e_{2}$ being the two extra edges of $K$. 
These two edges are incident to a common vertex, which we denote $v$.
Observe that this vertex must be contained in some 
endblock $B_{v}$ of $L$ (otherwise, $\eta \geqslant 3$).

Let $B_{f}$ be the block of $L$ including the edge $f$. 
Suppose $B_{f}$ is a cycle. This cycle has length at most $7$.
If $v \in V(B_{f})$, then each extra edge gives a diamond with at most $8$ edges in $K$.
Similarly, if $v \notin V(B_{f})$, then using the two extra edges we obtain a diamond
in $K$ with at most $9$ edges. Thus, in both cases we get
a contradiction. This implies $B_{f} \cong K_{2}$. In particular, 
$B_{f} \neq B_{v}$, since otherwise we would have $|B_{f}| \geqslant 3$.

Now, $B_{f}$ cannot be an endblock of $L$, as otherwise
some vertex of $B_{f}$ would have degree only $2$ in $K$. 
This implies that there is an endblock of $L$ distinct from $B_{v}$ that contains a vertex
having at most two neighbors in $K$, a contradiction.
This concludes the proof.
\end{proof}

\begin{lemma}
\label{lem:heart_of_const}
For the minimal hitting set $X$ output by Algorithm~\ref{algo:constant}, we have
$$
\sum_{v \in X} a_{i,v} \leqslant 9 \, \beta_i
$$
for all $i \in \{1,\ldots,k\}$.
\end{lemma}

\begin{proof}
If the $i$th inequality of the working LP relaxation is an extended sparsity inequality, the result follows from the previous lemma. Thus, we may assume that the $i$th inequality is deduced from a diamond $\wt{D}$ (thus $\beta_i = 1$). Hence, $\wt{D}$ has at most $2q-1 = 9$ edges. Let $S$ denote the support graph of the $i$th inequality. If $S$ is of type 1, then it has at most $9$ pieces. If it is of type 2, it has at most $8$ pieces. If it is of type 3 or 4, the left hand side is at most two. The result then follows from the arguments used in the proof of Lemma \ref{lem:heart_of_log(n)}.
\end{proof}

Our main result directly follows from Lemma \ref{lem:heart_of_const}. The proof is identical to the proof of Theorem \ref{Ologn} and hence is omitted.

\begin{theorem}
Algorithm~\ref{algo:constant} is a 9-approximation for the diamond hitting set problem.
\end{theorem}

% ======================================================================

\section*{Acknowledgements}

We thank Dirk Oliver Theis for his valuable input in the early stage of this research. We also thank Jean Cardinal and Marcin Kami\'nski for stimulating discussions.

% ======================================================================
\bibliography{diamondVScactus}
\bibliographystyle{plain}

\end{document}